\RequirePackage{fix-cm}
\documentclass[smallextended]{svjour3}       
\smartqed  
\usepackage{tikz}
\usepackage{multirow,multicol}
\usepackage [latin1]{inputenc}
\usepackage{graphicx}
\usepackage{amsmath,mathrsfs}
\usepackage{amssymb,cite}

\newtheorem{construction}{Construction}
\newtheorem{Claim}{Claim}

\newcommand{\F}{\mathbb {F}}

\newcommand{\vc}{{\bf c}}

\begin{document}

\title{Optimal Quaternary $(r,\delta)$-Locally Recoverable Codes: Their Structures and Complete Classification}

\author{Li Xu \and Zhengchun Zhou  \and Jun Zhang \and Sihem Mesnager}
\institute{Li Xu \at School of Mathematics, Southwest Jiaotong University, Chengdu, 611756, China. \\ \email{xuli1451@163.com}
\and Zhengchun Zhou \at School of Mathematics, Southwest Jiaotong University, Chengdu, 611756, China.\\
\email{zzc@swjtu.edu.cn}
\and Jun Zhang \at School of Mathematical Sciences, Capital Normal University, Beijing, 100048, China. \\ \email{junz@cnu.edu.cn}
\and Sihem Mesnager \at Department of Mathematics, University of Paris VIII, 93526 Saint-Denis, France, University of Paris XIII, CNRS, UMR 7539 LAGA, Sorbonne Paris Cit$\acute{\mathrm{e}}$, 93430 Villetaneuse, France, and Telecom Paris, Polytechnic Institute of Paris, 91120 Palaiseau, France. \\
\email{smesnager@univ-paris8.fr}
}

\date{Received: date / Accepted: date}
\maketitle

\begin{abstract}

Modern distributed and cloud storage systems have reached such a massive scale that recovery from several failures is now a part of the regular operation of the system rather than a rare exception. In addition, storage systems have to provide high data availability to ensure high performance. Redundancy and data encoding must be introduced into the system to address these requirements. Locally recoverable (LRC) codes have been introduced as a family of erasure codes that support the repair of a failed storage node by contacting a small number of other nodes in the cluster. Boosted by their applications in distributed storage, LRC codes have attracted a lot of attention in the recent literature since the introduction of the concept of codes with locality by Gopalan, Huang, Simitci, and Yekhanin in 2012. Since then, LRC codes and their variants have been extensively studied, and many exciting results regarding their locality properties (including a generalized Singleton bound involving the locality of the code). Also, constructions have been obtained, mainly using powerful algebraic coding theory techniques developed by Tamo and Barg in 2014 for constructing various families of LRC codes, including algebraic-geometric codes.

 Aiming to recover the data from several concurrent node failures, linear $r$-LRC codes with locality $r$ were extended into $(r, \delta)$-LRC codes with locality $(r, \delta)$ which can enable the local recovery of a failed node in case of more than one node failure. Optimal LRC codes are those whose parameters achieve the generalized Singleton bound with equality. In the present paper, we are interested in studying optimal LRC codes over small fields and, more precisely, over the finite field $\mathbb{F}_{4}$. Specifically, the $d-$optimal $(r, \delta)-$LRC codes considered in this article are in fact quaternary $[n,k,d]$-linear codes with locality $(r, \delta)$ which are simultaneously $r$-optimal and $d$-optimal (with minimum distance $d \geq \delta > 2$, and dimension $k > r \geq 1$). We adopt an approach by investigating optimal quaternary $(r,\delta)$-LRC codes through their parity-check and generator matrices. Our study includes determining the structural properties of optimal $(r,\delta)$-LRC codes, their constructions, and their complete classification over $\F_4$ by browsing all possible parameters. We emphasize that the precise structure of optimal quaternary $(r,\delta)$-LRC codes and their classification are obtained via the parity-check matrix, and generator matrix approaches. Besides, arguments involving projective space and related objects are intensively used in this paper to browse all the possible cases to obtain the desired constructions. Finally, compared to the recent literature, our structural and classification results about those optimal quaternary $(r,\delta)$-LRC codes over $\mathbb{F}_{4}$ are complete. Notably, we highlight that our proofs-techniques are different from those used recently for optimal binary and ternary $(r,\delta)$-LRC codes obtained by Hao et al. in 2017 and 2019, respectively.

\end{abstract}


\section{Introduction}

Distributed storage systems store data on several distributed nodes and are widely used in file system storage, ample database storage, backup file, cloud storage, etc.
The need for highly scalable and reliable extensive data storage systems is because of the explosive growth in data.
Distributed storage systems provide reliable access to data through redundancy spread over individually unreliable nodes, where the replication scheme and coding mechanism are two
widespread techniques for ensuring reliability. The replication scheme is very simple, but it will be
highly inefficient with data growth since its large storage overhead. They are namely two design objectives for a distributed storage system. The first one is to never lose user data in the event of crashes (or at least make it highly improbable). The second is to serve user requests with low latency despite some temporarily unavailable servers. Due to their theoretical appeal and the motivations for their applications in large-scale distributed storage systems, locally recoverable (LRC) codes have been introduced via the concept of locality introduced by P. Gopalan et al. \cite{Gopalan2012}. In recent years, they have been proliferated to bring and develop more robust solutions to those problems related to DSSs than replication schemes. Since their introduction, LRC codes have recently been an attractive subject in research in coding theory, where a single storage node erasure is considered a frequent error event. Such codes form a family of erasure codes that support the repair of a failed storage node by contacting a small number of other nodes in the cluster.

For LRC codes, if a symbol is lost due to a node failure, its value can be recovered if every coordinate of the codeword $\vc=(c_1,\ldots,c_n)\in \mathcal{C}$ can be recovered from a subset of $r$ other coordinates of $\vc$. Mathematically, it gives the following definition.

\begin{definition}[LRC codes]
A code $\mathcal {C}$ has locality $r$ if for every $i\in[n]:=\{1, 2, \cdots , n\}$ there exists a subset $R_i \subseteq [n]\setminus\{i\}$, $\# R_i\leq r$ and a function $\phi_i$ such that for every codeword $\vc\in\mathcal C$,
$$c_i=\phi_i(\{c_j, j\in R_i\}).$$
\end{definition}
An $(n,k,r)$ LRC code $\mathcal{C}$ over (the finite field of $q$ elements)  $\mathbb {F}_q$ is of code length $n$, cardinality $q^k$, and locality $r$. The parameters of an $(n,k,r)$ LRC code have been intensively studied.

\begin{theorem}[\cite{Gopalan2012,Papailiopoulos2012}]
Let $\mathcal {C}$ be an $(n,k,r)$ LRC code over $\mathbb {F}_q$, then
the minimum distance of $\mathcal {C}$ satisfies
\begin{equation}\label{r_Singleton}
d \leq n-k-\lceil\frac{k}{r}\rceil+2.
\end{equation}
\end{theorem}

Note that if $r=k$, the upper bound (\ref{r_Singleton}) coincides with the well-known Singleton bound $d\leq n-k+1$.

Given the above upper bound on the minimum distance, optimal LRC codes have been defined as follows.
\begin{definition}[Optimal LRC codes]
LRC codes for which $d= n-k-\lceil k/r \rceil+2$ are called optimal codes.
\end{definition}

Optimal LRC codes have been extensively studied in recent years (see. e.g., \cite{Tamo2014,Rawat2014,Tamo2016,Tamo2016B}), notably, exceptional results about their design using powerful algebraic techniques of coding theory were provided and developed by Tamo and Barg \cite{Tamo2014} in 2014 for constructing families of LRC codes including algebraic-geometric codes. LRC codes can be constructed in several ways, and many constructions have been proposed in the recent literature (see. e.g., \cite{jin2019constructions,cai2019optimal,wang2021construction,zhang2020locally}).
One of the most interesting constructions of optimal LRC codes is due to Tamo and Barg \cite{Tamo2014} realized via constructing polynomials (called good polynomials) of degree $r + 1$ which are constant on subsets of $\mathbb {F}_q $ of cardinality $r + 1$. Besides, all the possible classifications of optimal binary, ternary and quaternary LRC codes attaining the bound were presented by Hao et al. in \cite{Hao2020}, \cite{Hao2017} and \cite{Hao2021}, respectively. The original concept of locality only works when exactly one erasure occurs (that is, one node fails). Prakash et al. \cite{Prakash2012} introduced the concept of $(r, \delta)$-locality for linear codes aiming local recovery in case of more than one node failure. More precisely,

\begin{definition}[$(r,\delta)$-LRC codes]\label{LRC}
 The $i-$th code symbol $\textbf{c}_i$, $i \in [n]$, in an $[n, k, d]$ linear code $\mathcal C$, will be said to have locality $(r,\delta)$ if there exists a punctured subcode of $\mathcal C$ with support containing $i$, whose length $\leq r +\delta - 1$ and whose minimum distance $\geq \delta$. An $(r, \delta)$-LRC code is a linear code all of whose $n$ code symbols have locality $(r,\delta)$.
\end{definition}

It was proved in \cite{Prakash2012} that the minimum distance of an $(r, \delta)$-LRC codes is upper bounded by
$$
d \leq n-k+1-(\lceil\frac{k}{r}\rceil-1)(\delta-1).
$$

Optimal $(r, \delta)$-LRC codes with minimum distance meeting this bound given in \cite{Chen2018}, \cite{Ernvall2016}, \cite{Kamath2014} and \cite{Song2014}. In particular, Hao et al. proposed in \cite{Hao2017_2} and \cite{Hao2019} an approach based on analyzing parity-check matrices to study $(r, \delta)$-LRC codes in the binary and ternary cases with $\delta > 2$ and enumerated all possibilities for obtaining optimal binary and ternary $(r, \delta)$-LRC codes.

In this paper, we employ the parity-check matrix and generator matrix approaches to study the classifications of $(r, \delta)$-LRCs $(\delta > 2)$ over the quaternary field. Note that parity-check-matrix-based methods to study LRCs have been extensively used in literature. Many existing constructions of optimal LRC codes were given by presenting their explicit parity-check matrices, e.g.,\cite{Guruswami-et-al-2019,Jin2019,Prakash2012,Xing-Yuan-2018,Wang-et-al-2015,Hao2020}. The parity-check matrix can indeed characterize the locality property of an LRC code.
However, we highlight that the precise structure of optimal quaternary $(r,\delta)$-LRC codes and their classification are obtained use proofs-techniques pretty different from those used recently for optimal binary and ternary $(r,\delta)$-LRC codes obtained by Hao et al. in \cite{Hao2017_2} and \cite{Hao2019}. We combine parity-check matrix and generator matrix and use finite geometry as an essential tool in the classification process.

The paper is organized as follows. Section \ref{Sec:Preliminaries} fixes our notations and introduces the necessary background, and briefly recalls main results for general linear codes and the specific case of linear $(r,\delta)$-LRC codes that we need in the article.
Section \ref{Sec: LRC-prop} derives some crucial properties of optimal $4$-ary $(r, \delta)$-LRC codes. Next, analyze all possible parameters of optimal quaternary $(r, \delta)$-LRC codes entirely, and proposes explicit constructions of parity-check matrices for all these parameters. We shall precisely study in Sections \ref{Sec:LRC-$d=3$}, \ref{Sec:LRC-$d=4$} and \ref{Sec:LRC-$d=5$} optimal quaternary $(r, \delta)$-LRC codes with specific possible values of the minimum distance $d=3,~4,$ or $d\geq 5$ and derive explicitly the construction of the corresponding optimal $4$-ary $(r, \delta)$-LRC codes. Finally, Section \ref{Sec-Conclusions} concludes the article.

\section{Preliminaries}\label{Sec:Preliminaries}

Given a finite set $E$, $\# E$ will denote its cardinality. 
 $A^{T}$ denotes the transposed matrix of matrix $A$, $\text{Rank}(A)$ denotes the rank of $A$, $A \otimes B$ denotes the Kronecker product of matrices $A$ and $B$ and $I_m$ denotes the $m \times m$ identity matrix.

Let $\mathbb{F}_{q}$ be the finite field of order $q$ where $q$ is a prime power 
Furthermore, $\mathbb F_q^*$ be the multiplicative cyclic group of $\mathbb{F}_{q}$.
In this paper, we are interested in investigating LRC quaternary codes. Consequently, we will focus on the finite field $\mathbb{F}_4$ of characteristic $2$ consisting of the elements $\{0,1,w,w^2\}$, with $w^2=w+1$, $w^3=1$ (where $w$ denotes the primitive element in $\mathbb{F}_4$).
The \emph{Hamming weight}, denoted by wt$(\textbf{a})$, of a vector $ \textbf{a} \in \F_{q}^n$ is the cardinality of its \emph{support} defined as $Supp( \textbf{a})=\{i \in [n] : a_i\neq 0\}$ (that is, $wt(\textbf{a}):=\#Supp(\textbf{a})$), where $[n] = \{1, 2, \cdots , n\}$, $a_i$ is $i$-th component of $\textbf{a}$. An $[n, k, d]_q$ linear code $\mathcal C$ is a $k$-dimensional vector subspace of $\mathbb{F}_{q}$ with minimum Hamming distance $d$, where $d=\min_{\textbf{a},\textbf{b}\in\mathcal{C},\textbf{a}\neq \textbf{b}}d_H(\textbf{a},\textbf{b})$ in which $d_H$ denotes the Hamming distance
between vectors (called codewords), i.e., $d_H(\textbf{a},\textbf{b})=\#\{i \in [n] : a_i\not=b_i\}$. Usually, if the context is clear, we omit the subscript $q$ of $q$-ary linear code by convention in the sequel (we shall write $[n, k, d]$ instead of $[n, k, d]_q$).  

Let $\mathcal C$ be an $[n, k, d]_q$ linear code. Then its (Euclidean) dual code is denoted by $\mathcal C^ {\perp}$ and defined as
$$\mathcal C^ {\perp}=\{ (b_1, b_2, \ldots, b_{n})\in \mathbb F_{q}^n: \sum _{i=1}^{n} b_i  c_i=0, \forall (c_1, c_2, \ldots, c_{n}) \in \mathcal C \}.$$

The $k\times n$ generator matrix $G$ and $(n-k)\times n$ parity-check matrix  $H$ of code $\mathcal C$ are respectively composed of $k$ linearly independent codewords in $\mathcal C$ and $n-k$ linearly independent codewords in $C^{\bot}$. They satisfy  $GH^{T}=0$. It is well-known that a linear code has minimum distance $d$ if and only if its parity check matrix has a set of $d$ linear dependent columns but no set of $d-1$ linear dependent columns.


For an $[n, k, d]_q$ linear code $\mathcal C$, the classical Singleton Bound says that $d \leq n - k + 1$. A code for which equality holds in the Singleton Bound is called maximum distance separable, abbreviated \emph{MDS}. If $\mathcal C$ is an $[n,k,n-k+1]$ MDS code, then $C^{\bot}$ is an $[n, n-k, k+1]$ MDS code (see. e.g. \cite{Huffman2003}). MDS codes are desirable for distributed storage applications because of their optimal storage versus reliability tradeoff.


The linear codes with parameters $[n, 1, n]$ or $[n, n-1, 2]$ are trivial MDS codes. For a $q$-ary non-trivial MDS code with dimension $2 \leq k \leq n - 2$, the following result holds.

\begin{lemma}[\cite{MacWilliams1977}] \label{MDS} Let $\mathcal C$ be a $q$-ary $[n, k, n-k+1]$ MDS code. If $k \geq 2$, then $q \geq n - k + 1$. If $k \leq n - 2$, then $q \geq k + 1$.
\end{lemma}

By the above lemma, a non-trivial $q$-ary MDS code has dimension $2 \leq k \leq q - 1$ and length $n \leq q + k - 1$. Let $q=4$, the following proposition gives the parameters of all possible quaternary MDS codes.
\begin{proposition} \label{MDS4} The code parameters of all possible quaternary MDS codes are
$$[4, 2, 3],~ [5, 2, 4],~ [5, 3, 3],~[6, 3, 4],~[n, 1, n], ~\mbox{and}~~[n, n-1, 2], ~n \geq 2.$$
\end{proposition}

We recall important equivalence notions of codes over the finite field $\mathbb F_q$ (see \cite{Huffman2003} Sections 1.6 and 1.7).

\begin{definition}
 Let $\mathcal C$ and $\mathcal C^\prime$ be two linear codes of the same length over $\mathbb F_q$. $\mathcal C$ and $\mathcal C^\prime$ are \emph{permutation equivalent} if there is a permutation matrix $P$ such that $G$ is a generator matrix of $\mathcal C$ if and only if $GP$ is a generator matrix of $\mathcal C^\prime$.
\end{definition}

Recall that a monomial matrix is a square matrix with exactly one nonzero entry in each row and column. A monomial matrix M can be written either in the form of $DP$ or the form of $PD^\prime$ where $D$ and $D^\prime$ are diagonal matrices, and $P$ is a permutation matrix.

\begin{definition}
 Let $\mathcal C$ and $\mathcal C'$ be two linear codes of the same length over $\mathbb F_q$, and let $G$ be a generator matrix of $\mathcal C$. Then $\mathcal C$ and $\mathcal C'$ are \emph{monomially equivalent} if there is a monomial matrix $M$ such that $GM$ is a generator matrix of $\mathcal C^\prime$.
\end{definition}

\begin{definition}[Equivalence of codes]
Two codes $\mathcal C$ and $\mathcal C'$ in $\mathbb F_q^n$ are called \emph{equivalent} if $\mathcal C'=\sigma(\mathcal C_{\mathbf a})$ for some permutation $\sigma$ of $\{1,2,\ldots, n\}$ and $\mathbf a \in \mathbb (\mathbb F_q^*)^n$. Any $[n,k]$ linear code over a finite field is equivalent to a code generated by a matrix of the form $[I_k: P]$ where $I_k$ denotes the $k\times k$ identity matrix.
\end{definition}



Notably, the quaternary code with parameters $[6, 3, 4]$ is often called the \emph{Hexacode}. In the sense of equivalence of codes, this code is unique (see. e.g., \cite{Huffman2003}).
In this paper, we shall choose $G$ as the generator matrix of the Hexacode,

$$G=\left(\begin{array}{llllll}
1 & 0 & 0 & 1 & 1 & 1\\
0 & 1 & 0 & 1 & w & w^2\\
0 & 0 & 1 & 1  & w^2 & w
\end{array}\right).$$

In general, we can \emph{puncture} a linear code $\mathcal C$ on the coordinate set $S$ by deleting components indexed by the set $S$ in all codewords of $\mathcal C$ and denote $C^S$ the resulting code.  Consider the set $\mathcal{C}(S)= \{\vc \in \mathcal{C} : c_i=0,~i\in S\}$, this set is a subcode of $\mathcal C$.
Puncturing $\mathcal{C}(S)$ on $S$ gives a code of length ($n-\#S$) called the code \emph{shortened} on $S$ and denoted $\mathcal C_S$. The dual of a punctured code is a shortened code, that is, $(C^S)^{\perp}=(\mathcal C^{\bot})_S$ and $(\mathcal C_S)^{\bot} = (\mathcal C^{\bot})^S$ (see. e.g. \cite{Huffman2003}).

We now recall the crucial notion of projective space and related objects and properties, which will be intensively used in this paper. First, we can define the projective space $PG(m-1, \F_q)$ from the $k$-dimensional vector space $V_m(\F_q)$ over $\F_q$ in the following way. The \emph{points} of $PG(m-1, \F_q)$ are the one-dimensional subspaces of $V_m(\F_q)$, the \emph{lines} of $PG(m-1, \F_q)$ are the two-dimensional subspaces of $V_m(\F_q)$, the \emph{planes} of $PG(m-1, \F_q)$ are the three-dimensional subspaces of $V_m(\F_q)$.
We can think of the subspace of projective space as a collection of the points it contains, and its intersection determines the intersection of subspaces in the vector space. We also have the following lemmas about the number of subspaces of dimension $i$ in $V_m(\F_q)$.

\begin{lemma}[\cite{Projectivespace}] \label{i-space}
The number of $i$-dimensional subspaces of $V_m(\F_q)$ is
$$
\frac{\left(q^{m}-1\right)\left(q^{m-1}-1\right) \cdots\left(q^{m-i+1}-1\right)}{\left(q^{i}-1\right)\left(q^{i-1}-1\right) \cdots(q-1)}.
$$

\end{lemma}

\begin{lemma}[\cite{Projectivespace}] \label{i-j-space}
The number of $i$-dimensional subspaces of $V_m(\F_q)$ containing a fixed $j$-dimensional subspace is equal to the number of $(i - j)$-dimensional subspaces in $V_{m-j}(\F_q)$. This number equals
$$
\frac{\left(q^{m-j}-1\right)\left(q^{m-j-1}-1\right) \cdots\left(q^{m-i+1}-1\right)}{\left(q^{i-j}-1\right)\left(q^{i-j-1}-1\right) \cdots(q-1)}.
$$
\end{lemma}

\section{On $(r,\delta)$-Locally Repairable Codes}\label{Sec: LRC-prop}
\subsection{Generalities}

A linear code $\mathcal C$ is a LRC code with locality $r$ if for any $i\in[n]$, there exists a subset
 $\mathcal R_i\subseteq [n]\backslash\{i\}$ with $\#\mathcal R_i\leq r$ such that the $i$-th symbol $\textbf{c}_i$ can be recovered
 by $\{\textbf{c}_j\}_{j\in \mathcal R_i}$.

 When an LRC code with locality $r$ much less than $k$ is employed, the repair cost is lower than MDS codes, as only a small number of storage nodes are involved in the repair process.

A set $\mathcal R_i$ is called a recovery or repair set for $\textbf{c}_i$.
Furthermore, if for any $i\in[n]$, there is a repair sets of size at most $r$ even if $\delta -2$ other symbols fail, we refer to such a code as an $(r,\delta)$-LRC code.

 In order to maximize the reliability of storage systems, it is desirable to obtain codes where lost data can be repaired by contacting a small number of nodes $r$.

\subsection{Equivalent concepts of $(r, \delta)$-locality and properties}

The concept of $(r, \delta)$-locality can be reformulated mathematically with different equivalent ways as follows.

\begin{itemize}

\item $(r,\delta)$-LRC codes by puncturing:

a code $\mathcal C$ is an $(r,\delta)$-LRC code if for any $i \in [n]$ there exists a subset $\mathcal R_i \subseteq [n]$ such that
 \begin{enumerate}
\item [(1)] $i \in \mathcal R_i$ and $\#\mathcal R_i \leq r + \delta - 1$;
\item [(2)] the minimum distance of the code $\mathcal{C}|_{\mathcal R_i}$ is at least $\delta$,
 \end{enumerate}
where $\mathcal{C}|_{\mathcal R_i}$ is obtained by puncturing $\mathcal C$ on the coordinates set $[n]\setminus \mathcal R_i$.

\item $(r,\delta)$-LRC codes by considering the parity-check matrix:

a code $\mathcal C$ is an $(r,\delta)$-LRC code if for any $i \in [n]$ the parity-check matrix $H$ of $\mathcal C$ contains a ($v_i \times n$) submatrix $H_i$, $1 \leq v_i \leq n - k$, such that
 \begin{enumerate}
\item [(1)] $H_i$ has support $\mathcal R_i$;
\item [(2)] any $\delta - 1$ columns of $H_i$ with indices drawn from $\mathcal R_i$ are linearly independent.
 \end{enumerate}

 \end{itemize}
 Let $H_{i}'$ denote the submatrix of $H_i$ consisting of all nonzero columns of $H_i$, then the parameters of linear code $\mathcal{C}|_{\mathcal R_i}=\mathcal C^{[n]\setminus \mathcal R_i}$ is
  $$[\#\mathcal R_i \leq r + \delta - 1,~ k_i \leq r, ~d_i \geq \delta],$$
  and the parameters of linear code $\langle H_{i}' \rangle=(\mathcal C^{\bot})_{[n]\setminus \mathcal R_i}$ is
   $$[\#\mathcal R_i \leq r + \delta - 1, ~\geq \delta-1,~ \leq r+1].$$

Let $\mathcal C$ be an $(r,\delta)$-LRC code, we can find a full-rank parity-check matrix $H$ of $\mathcal C$ that is divided into $l+1$ parts
\begin{equation}\label{h}
H=\left[\begin{array}{c}H_1\\H_2\\...\\H_l\\H^*\end{array}\right].
\end{equation}
The rows in $H_i$, $i\in [l]$, are called \emph{local rows}, $H_i$ is said to be a \emph{local group}; the rows in $H^*$ are called \emph{global rows}, and $H^*$ is called the \emph{global group}.

Let $S_i$ be the support of local group $H_i$, $i\in[l]$, then
 \begin{enumerate}
\item [(1)] $\#S_i \leq r+\delta -1$;
\item [(2)] The union of the supports of $l$ local groups $\#(S_1\cup ... \cup S_l) = n$, and any $l-1$ local groups can not cover all the $n$ coordinates.
  \end{enumerate}

Note that, the number of rows in each local group $\geq \delta-1$, we get the following proposition.
\begin{proposition}\label{l}
If $\mathcal C$ is an $(r,\delta)$-LRC code, then $l$ is in the following range
\begin{equation}\label{ll}
\lceil\frac{k}{r}\rceil \leq l \leq \lfloor \frac{n-k}{\delta - 1}\rfloor.
\end{equation}
\end{proposition}

The following generalization of the Singleton bound for LRC codes was among others proven in \cite{Kamath2014} (see. [Theorem 3.1]), \cite{Tamo2014} (see. [Construction 8] and [Theorem 5.4]), and \cite{Prakash2012} (see.[Theorem 2]).

\begin{theorem} [Singleton-like Bound] \label{$d$-optimal}
The minimum distance $d$ of an $[n,k,d]_q$ linear $(r, \delta)$-LRC code $\mathcal C$ is upper bounded by
\begin{equation}\label{rd_Singleton}
d \leq n-k+1-(\lceil\frac{k}{r}\rceil-1)(\delta-1).
\end{equation}
\end{theorem}

For $\delta=2$ and $r=k$ it coincides with the classical standard Singleton bound. Throughout this contribution we call a code \emph{$d$-optimal} if its minimum distance  meets the bound (\ref{rd_Singleton}) with equality. The parameter $r$ is related to the repair cost of locally repairable code. The smaller $r$ means fewer storage nodes need to be downloaded to repair the failed storage nodes.
If $\mathcal C$ is an $(r, \delta)-$LRC code and we cannot decrease the value of $r$ while the other parameters remain unchanged, then we say that $\mathcal C$ is \emph{$r$-optimal}.
In this paper, we mainly focus on the $r-$optimal $(r, \delta)-$LRC codes, that is, all the optimal $(r, \delta)-$LRC codes are both $r$-optimal and $d$-optimal, with minimum distance  $d \geq \delta > 2$, and dimension $k > r \geq 1$.



The following lemma gives an important property of optimal $(r, \delta)-$LRC code. This property is proved in \cite{Hao2017_2}, we still present a brief proof for completeness.

\begin{lemma}[\cite{Hao2017_2}]\label{H'}
 Let $\mathcal C$ be an optimal $(r,\delta)-$LRC code, and the parity-check matrix is shown in (\ref{h}).
 Pick any $\lceil\frac{k}{r}\rceil-1$ local groups in $H$, let $H'$ be the $m' \times n'$ matrix obtained from $H$ by deleting the rows in these $\lceil\frac{k}{r}\rceil-1$ groups and all the columns whose coordinates are covered by the supports of these groups. Then $H'$ has full rank and the $[n',k',d']$ linear code $\mathcal C^\prime$ defined by the parity-check matrix $H'$ is an MDS code with $d' = d$.
 \end{lemma}

\begin{proof}
The number of the rows in $H'$ is $$ m' \leq n-k-(\lceil\frac{k}{r}\rceil-1)(\delta-1),$$
and the number of the columns is $$ n' \geq n-(\lceil\frac{k}{r}\rceil-1)(r+\delta-1).$$
Since $r(\lceil\frac{k}{r}\rceil-1) < k$, $ n'>m'$, we have $\text{Rank}(H') \leq m'$.

Since $H'$ is a parity-check matrix of the linear code $\mathcal C^\prime$, we can get
\begin{equation}\label{d1}
d' \leq n'-k'+1 =\text{Rank}(H')+1 \leq m'+1 \leq n-k-(\lceil\frac{k}{r}\rceil-1)(\delta-1)+1.
\end{equation}
Because $\mathcal C^\prime$ is a shortened code of $\mathcal C$,
\begin{equation}\label{d2}
d'\geq d = n-k+1-(\lceil\frac{k}{r}\rceil-1)(\delta-1).
\end{equation}
Combining inequalities (\ref{d1}) and (\ref{d2}), we have
$$d' = n'-k'+1 = m'+1 = n-k-(\lceil\frac{k}{r}\rceil-1)(\delta-1)+1=d,$$
which implies that $\mathcal C^\prime$ is an MDS code with $d' = d$ and $\text{Rank}(H') = m' = n-k-(\lceil\frac{k}{r}\rceil-1)(\delta-1)$.
\end{proof}

Let $\mathcal C$ be an $[n, k]_q$ linear codes with $(r, \delta)-$locality ($\delta > 2$) and optimal minimum distance $d = n-k+1-(\lceil\frac{k}{r}\rceil-1)(\delta-1)$, the parity-check matrix $H$ is shown in (\ref{h}). There are some significant and practical conclusions about code $\mathcal C$.

\begin{corollary} \label{Hi}
  Let $\mathcal C$ be an optimal $(r, \delta)-$LRC code. Then there are exactly $\delta-1$ rows in each local group $H_i$, $i\in [l]$.
 \end{corollary}

\begin{proof}
From the proof of Lemma \ref{H'}, we have $m' = n-k-(\lceil\frac{k}{r}\rceil-1)(\delta-1)$, which means that the total number of the rows in the $\lceil\frac{k}{r}\rceil-1$ removed local groups is $n-k-m'= (\lceil\frac{k}{r}\rceil-1)(\delta-1)$. Since each local group contains at least $\delta-1$ rows, we get each of the $\lceil\frac{k}{r}\rceil-1$ groups contains exactly $\delta-1$ rows. Because the removed groups are randomly selected, each of the $l$ groups in $H$ contains exactly $\delta-1$ rows.
\end{proof}

\begin{corollary}\label{HMDS} Let $\mathcal C$ be an optimal $(r, \delta)-$LRC code. Then for any $i\in [l]$, the code $\langle H_{i}' \rangle $ is an $[\#S_i \leq r+\delta -1, \delta-1, \#S_i-\delta+2]$ MDS code, that is, the punctured code $\mathcal{C}|_{S_i}$ of $\mathcal C$ is an $[\#S_i \leq r+\delta -1, \#S_i-\delta+1, \delta]$ MDS code.
 \end{corollary}

\begin{corollary}
 \label{delta} If $\mathcal C$ is an $r-$optimal $(r, \delta)-$LRC code with minimum distance $d$ attaining the Singleton-like bound. Then there is an $i\in [l]$, such that the code $\langle H_{i}' \rangle $ is an $[r+\delta -1, \delta-1, r+1]$ MDS, that is, the code $\mathcal{C}|_{S_i}$ is an $[ r+\delta -1, r, \delta]$ MDS code.
 \end{corollary}
\begin{proof}
 By Corollary \ref{HMDS}, we only need to prove that there is a local group of length $r+\delta -1$.
 Assuming that the length of each local group is less than $r+\delta -1$, then for any $i\in [l]$, the code $\mathcal{C}|_{S_i}$ is an $[\#S_i < r+\delta -1, \#S_i-\delta+1, \delta]$ MDS code. We can find an $r_i < r$ such that $n_i = r_i+\delta -1$, and $\mathcal{C}|_{S_i}$ is an $[\#S_i = r_i+\delta -1, r_i, \delta]$ MDS code. This implies that the $i-$th code symbol $\textbf{c}_i$ has locality $(r_i, \delta)$. Further, we can get that $\mathcal C$ has locality $(r', \delta)$ for some $r'<r$, this contradicts that $\mathcal C$ is $r-$optimal.
\end{proof}

\begin{lemma}[\cite{Song2014}]\label{k=rs}
Assume $\mathcal C$ is an optimal $(r, \delta)$-LRC code. If $r\mid k$ and $r < k$, then the following conditions hold:
 \begin{enumerate}
\item [(1)] $S_1$,..., $S_l$ are mutually disjoint;
\item [(2)] for any $i \in [l]$, $\#S_i = r + \delta - 1$ and the punctured code $\mathcal{C}|_{S_i}$ is an $[r + \delta - 1, r, \delta]$ MDS code.
 \end{enumerate}
In particular, we have $(r + \delta - 1) \mid  n$.
\end{lemma}

By the above lemma, and up to a rearrangement of the code coordinates, an optimal $(r, \delta)$-LRC code with $r\mid k$ has parity-check matrix $H$ of the following form:
$$
H=\left[\begin{array}{ccc|ccc}
I_{\delta-1} &       &             & Q_{1}  &        & \\
             &\ddots &             &        & \ddots & \\
             &       & I_{\delta-1}&        &        &Q_{l} \\
\hline
   &0_{(lr-k)\times l(\delta-1)}&   & A_{1}  &\cdots  & A_{l}
\end{array}\right]
$$
where the matrices $Q_i$ and $A_i$ appearing above are of sizes $(\delta - 1) \times r$ and $(lr-k)\times r$, respectively.

The following lemma reveals the relation between the minimum distance  $d$ of optimal $(r, \delta)$-LRC codes and the order $q$ of the finite field.

\begin{lemma} [\cite{Hao2019}] \label{distance} Let $\mathcal C$ be an optimal $q$-ary $(n, k, r, \delta)$-LRC with minimum distance $d$ and dimension $k > r \geq 1$, then one has
$$d \leq \begin{cases}q, & \text { if } r \nmid(k-1), \\ \delta q, & \text { if } r \mid(k-1).\end{cases}$$
\end{lemma}

From Proposition \ref{MDS4} and Lemma \ref{H'},  we can distinguish only three possibilities  for the matrix $H'$.

\begin{proposition} \label{H'MDS} Let $\mathcal C$ be an optimal quaternary $(r, \delta)-$LRC code with minimum distance  $d$.
 \begin{enumerate}
\item [(1)]If $d = 3$, then $H'$ contains two rows and $H'$ is a parity-check matrix of a quaternary $[3, 1, 3]$, $[4, 2, 3]$ or $[5, 3, 3]$ MDS code.
\item [(2)]If $d =4 $, then $H'$ contains three rows and $H'$ is a parity-check matrix of a quaternary $[4, 1, 4]$, $[5, 2, 4]$, or $[6, 3, 4]$ MDS code.
\item [(3)]If $d \geq 5$, then $H'$ has more than three rows and $H'$ is a parity-check matrix of a quaternary $[d, 1, d]$ MDS code.
 \end{enumerate}
\end{proposition}

\section{Analyzing optimal quaternary $(r, \delta)$-LRC codes with minimum distance $d=3$}\label{Sec:LRC-$d=3$}
In this section, we assume that $\mathcal C$ is an optimal quaternary $(r, \delta)$-LRC code with minimum distance $d = 3$.
In this case,
$$n-k=d-1+(\lceil\frac{k}{r}\rceil-1)(\delta-1)=(\lceil\frac{k}{r}\rceil-1)(\delta-1)+2.$$
From (\ref{ll}), the number of local groups in $H$ is
$$\lceil\frac{k}{r}\rceil \leq l \leq \lceil\frac{k}{r}\rceil -1 +\lfloor \frac{2}{\delta - 1}\rfloor.$$
Hence, $\delta - 1 \leq 2$. By $\delta > 2$, we have $\delta = 3$. Furthermore, with the above inequality we get $l=\lceil\frac{k}{r}\rceil$,
$n-k=2\lceil\frac{k}{r}\rceil=2l$.

By Corollary \ref{Hi}, in the parity-check matrix $H$, each local group $H_i$, $i\in [l]$, contains precisely two rows, and there is no global group.

Let $k = rs + t$ , where $s\geq 1$, $0 \leq t \leq r - 1$. In Lemma \ref{H'}, the number of the columns covered by the $\lceil\frac{k}{r}\rceil-1$ removed local groups is $\gamma \leq (\lceil\frac{k}{r}\rceil-1)(r + \delta - 1)$. Using Proposition \ref{H'MDS}(1), we get
\begin{equation}\label{n'1}
n- (\lceil\frac{k}{r}\rceil-1)(r + \delta - 1)\leq n' \leq 5,
\end{equation}
therefore, $t=0, r\leq3$ or $0< t\leq 3$.

\subsection{Case  $t=3$.}\label{3.1}

When $t=3$, from inequality (\ref{n'1}), the length of $H'$ is
$$n'=5=n- (\lceil\frac{k}{r}\rceil-1)(r + \delta - 1).$$
This implies that these $\lceil\frac{k}{r}\rceil-1=s$ removed local groups have uniform support size $\#S_j=r + \delta - 1$. Since these $s$ groups can be chosen arbitrarily, all local groups have supports size $\#S_i=r + \delta - 1$, $i\in[l]$. By Corollary \ref{HMDS}, the codes $\langle H_{i}' \rangle $, $i\in[l]$, are $[r+2, 2, r+1]$ MDS codes. Since $l=\lceil\frac{k}{r}\rceil$, the remaining matrix $H'$ is a submatrix of $ H_{i}'$, we have $r+2\geq n' =5$, all $\langle H_{i}' \rangle $ are $[5, 2, 4]$ MDS codes. Thus, $r = 3$. This contradicts  the assumption $t = 3 \leq r - 1$. Hence, there is no optimal $(r, \delta)$-LRC code in this case.

\subsection{Case  $t=2$.}\label{3.2}

When $t=2$, the length of $H'$ is
$$4 \leq n' \leq 5.$$

\subsubsection{If $n'= 5$.}
$\langle H'\rangle$ is a quaternary $[5,2,4]$ MDS code. We have
$$n-n'= (\lceil\frac{k}{r}\rceil-1)(r + \delta - 1)-1,$$
which implies that in the $\lceil\frac{k}{r}\rceil-1=s$ removed local groups, there are two groups whose supports intersect at most one coordinate, and the rest groups have disjoint supports.

The $s$ subcodes $\langle H_{i}' \rangle $ corresponding to the removed local groups $H_i$ are either all $[r+2, 2, r+1]$ MDS codes and there are two groups of them whose supports intersect at one coordinate,
or one of them is $[r+1, 2, r]$ MDS code, and $s-1$ of them are $[r+2, 2, r+1]$ MDS codes with disjoint supports. From Proposition \ref{MDS4}, $n' \leq r+2 \leq 5$, $r+2=5$.

The code $\langle H_{j}' \rangle $ corresponding to the remaining local group $H_j$ is an $[\#S_j, 2, \#S_j-1]$ MDS code with length $\#S_j \geq n'= 5$, then $\#S_j = 5$, $\langle H_{j}' \rangle = \langle H'\rangle$ is a quaternary $[5,2,4]$ MDS code.

\subsubsection{If $n'= 4$.}
$\langle H'\rangle$ is a quaternary $[4,2,3]$ MDS code,
$$n-n'= (\lceil\frac{k}{r}\rceil-1)(r + \delta - 1),$$
 the supports of the $s$ removed groups are disjoint and each has support size exactly $r + \delta - 1$.

The $s$ subcodes $\langle H_{i}' \rangle $ corresponding to the removed groups $H_i$ are all $[r+2, 2, r+1]$ MDS codes, similar to above, $r+2=5$.

The code $\langle H_{j}' \rangle $ corresponding to the remaining group $H_j$ is an $[\#S_j, 2, \#S_j-1]$ MDS code with length $\#S_j \geq n'= 4$, then $\#S_j =4~\mbox{or}~5$, $\langle H_{j}' \rangle$ is a quaternary $[5,2,4]$ or $[4,2,3]$ MDS code.

In this case, the optimal $(r, \delta)$-LRC code has parameters
\begin{equation}\label{1}
n = 5l -1,~ k = 3l-1,~ d=3,~ r = 3, ~\delta = 3,~(l \geq 2).
\end{equation}

\begin{construction}\label{construction1}


The following parity-check matrix $H$ gives the optimal quaternary $(r, \delta)$-LRC code for the parameters in (\ref{1}).

$$
H=\left(\begin{array}{c|c}
\tilde{H} & \mathbf{0}_{4 \times(5(l-2))} \\
\hline \mathbf{0}_{(2(l-2)) \times 9} & I_{l-2} \otimes\left(\begin{array}{lllll}
1 & 0 & 1 & 1 & 1\\
0 & 1 & 1 & w & w^2
\end{array}\right)
\end{array}\right),
$$
where
$$
\tilde{H}=\left(\begin{array}{ccccccccc}
1 & 0 & 1 & 1 & \alpha & 0  & 0 & 0 & 0\\
0 & 1 & 1 & w & \beta  & 0  & 0 & 0 & 0\\
0 & 0 & 0 & 0 & 1      & 0  & 1 & 1 & 1\\
0 & 0 & 0 & 0 & 0      & 1  & 1 & w & w^2
\end{array}\right),
$$
and $(\alpha,\beta)^{T}=(0,0)^{T}$ or $(1, w^2)^{T}$.
\end{construction}

\subsection{Case $t=1$.}\label{3.3}
When $t=1$, the length of $H'$ is
$$3 \leq n' \leq 5.$$
By $r-1 \geq t$, we have $r \geq 2.$

\subsubsection{If $r=2$.}\label{3.3.1}

By $n' \leq r+ \delta -1$, $r+ \delta -1=4$, we have $n'=3~ \mbox{or}~ 4$.
Similar to case \ref{3.2}, the codes $\langle H_{i}' \rangle$, $i\in[l]$, are either all $[4, 2, 3]$ MDS code and there are two groups in which their supports intersect at one coordinate, or one of them is $[3, 2, 2]$ MDS code, and $l-1$ of them are $[4, 2, 3]$ MDS codes with disjoint supports.

In this case, the optimal $(r, \delta)$-LRC code has   the following parameters
\begin{equation}\label{2}
n = 4l -1,~ k = 2l-1,~ d=3,~ r = 2, ~\delta = 3,~(l \geq 2).
\end{equation}


\begin{construction}\label{construction2}

 By puncturing exactly one column from each local group of the matrix in Construction \ref{construction1}, we can get the parity-check matrix $H$ of the optimal quaternary $(r, \delta)$-LRC code for the parameters in (\ref{2}), as follows.

$$
H=\left(\begin{array}{c|c}
\tilde{H} & \mathbf{0}_{4 \times(4(l-2))} \\
\hline \mathbf{0}_{(2(l-2)) \times 7} & I_{l-2} \otimes\left(\begin{array}{llll}
1 & 0 & 1  & 1\\
0 & 1 & 1  & w^2
\end{array}\right)
\end{array}\right),
$$
where
$$
\tilde{H}=\left(\begin{array}{ccccccc}
1 & 0 & 1  & \alpha & 0  & 0 &  0\\
0 & 1 & 1  & \beta  & 0  & 0 &  0\\
0 & 0 & 0  & 1      & 0  & 1 &  1\\
0 & 0 & 0  & 0      & 1  & 1 &  w^2
\end{array}\right),
$$
and $(\alpha,\beta)^{T}=(0,0)^{T}$ or $(1, w^2)^{T}$.

\end{construction}

\subsubsection{If $r=3$.}\label{3.3.2}

By $n' \leq r+ \delta -1=5$, we have $n'=3, 4 ~\mbox{or} ~5$.

In this case, the optimal $(r, \delta)$-LRC code has parameters
\begin{equation}\label{3}
n = 5l -2,~ k = 3l-2,~ d=3,~ r = 3, ~\delta = 3,~(l \geq 2).
\end{equation}

Comparing the parameters in (\ref{3})  with the ones in (\ref{1}), it is easy to get the following construction.

\begin{construction}\label{construction3}

 By puncturing the first column from the matrix $H$ in Construction \ref{construction1}, we can get the parity-check matrix of the optimal quaternary $(r, \delta)$-LRC code for the parameters in (\ref{3}).
\end{construction}

\subsubsection{If $r\geq 4$.}\label{3.3.3}

In this case, $r+ \delta -1 \geq 6$. From Corollary \ref{delta}, there is an $ i\in [l]$, such that the code $\langle H_{i}' \rangle $ is an $[\tilde{n}\geq 6, 2, \tilde{n}+1]$ MDS code, this contradicts Proposition \ref{MDS4}. Therefore, there is no optimal $(r, \delta)$-LRC code in this case.

\subsection{Case $t=0, ~r\leq3$.}\label{3.4}

In this case, we have $r \mid k$. From Lemma \ref{k=rs}, the codes $\langle H_{i}' \rangle$, $i\in [l]$, are all $[r+2, 2, r+1]$ MDS codes with disjoint supports.

The optimal $(r, \delta)$-LRC code has parameters
\begin{equation}\label{4}
n = l(r+2),~ k = rl,~ d=3,~ 1 \leq r \leq 3, ~\delta = 3,~(l \geq 2).
\end{equation}

\begin{construction}\label{construction4}
 For the parameters in (\ref{4}), the following parity-check matrix $H$ gives the optimal quaternary $(r, \delta)$-LRC code with $r=3$.
By puncturing precisely one or two columns from each local group of the matrix $H$, we can obtain the parity-check matrix of the optimal quaternary $(r, \delta)$-LRC code with $r=2$ or $1$, respectively.

$$
H=\left(\begin{array}{c}
 I_{l} \otimes\left(\begin{array}{lllll}
1 & 0 & 1 & 1 & 1\\
0 & 1 & 1 & w & w^2
\end{array}\right)
\end{array}\right).
$$
\end{construction}

\section{Analyzing optimal quaternary $(r, \delta)$-LRC codes with minimum distance $d=4$}\label{Sec:LRC-$d=4$}
In this section, we assume that $\mathcal C$ is an optimal quaternary $(r, \delta)$-LRC code with minimum distance $d = 4$.
In this case,
$$n-k=d-1+(\lceil\frac{k}{r}\rceil-1)(\delta-1)=(\lceil\frac{k}{r}\rceil-1)(\delta-1)+3.$$
From (\ref{h}), the number of local groups in $H$ is
$$\lceil\frac{k}{r}\rceil \leq l \leq \lceil\frac{k}{r}\rceil -1 +\lfloor \frac{3}{\delta - 1}\rfloor.$$
Hence, $\delta - 1 \leq 3$. By $\delta > 2$, we have $\delta = 3$ or $4$. Furthermore, with the above inequality we can also get $l=\lceil\frac{k}{r}\rceil$,
$n-k=\lceil\frac{k}{r}\rceil(\delta -1) - \delta +4$.

By Corollary \ref{Hi}, if $\delta =3$, then in the parity-check matrix $H$ of an optimal $(r, \delta)$-LRC code, each local group $H_i$, $i\in [l]$, contains exactly two rows, the global group $H^{*}$ contains exactly one row. If $\delta =4$, then each local group $H_i$, $i\in [l]$, contains exactly three rows, and there is no the global group.

Similarly to inequality (\ref{n'1}), from Proposition \ref{H'MDS}(2), we get
\begin{equation}\label{n'2}
n- (\lceil\frac{k}{r}\rceil-1)(r + \delta - 1)\leq n' \leq 6,
\end{equation}
therefore, $t=0, r\leq3$ or $0< t\leq 3$.

\subsection{Case $t=3$.}\label{4.1}

When $t=3$, from inequality (\ref{n'2}), the length of $H'$ is
$$n'=6=n- (\lceil\frac{k}{r}\rceil-1)(r + \delta - 1),$$
which implies that these removed $\lceil\frac{k}{r}\rceil-1=s$ local groups have uniform support size $\#S_j=r + \delta - 1$. Since these $s$ groups can be chosen arbitrarily, all local groups $H_i$, $i\in [l]$, have supports size $\#S_i=r + \delta - 1$, and the codes $\langle H_{i}' \rangle $ are $[r + \delta - 1, \delta - 1, r+1]$ MDS codes. Since $l=\lceil\frac{k}{r}\rceil$, the length of remaining matrix $H'$ does not exceed the length of $ H_{i}'$, we have $r + \delta - 1\geq n' =6$. Because $\delta > 2$, from Proposition \ref{MDS4}, we can get that all $\langle H_{i}' \rangle $ are $[6, 3, 4]$ MDS codes. Thus, $r = 3$. This contradicts the assumption $t = 3 \leq r - 1$. Hence, there is no optimal $(r, \delta)$-LRC code in this case.

\subsection{Case $t=2$.}\label{4.2}

When $t=2$, the length of $H'$ is
$$5 \leq n' \leq 6.$$

\subsubsection{If $n'= 6$.}\label{4.2.1}
$\mathcal C^\prime$ is a quaternary $[6,3,4]$ MDS code, $\langle H'\rangle$ is a quaternary $[6,3,4]$ MDS code. We have
$$n-n'= (\lceil\frac{k}{r}\rceil-1)(r + \delta - 1)-1,$$
that is, in the $\lceil\frac{k}{r}\rceil-1=s$ removed groups, there are two groups whose supports intersect at most one coordinate, and the rest groups have disjoint supports.

These $s$ codes $\langle H_{i}' \rangle $ corresponding to the removed groups $H_i$ are either all $[r + \delta - 1, \delta - 1, r+1]$ MDS codes and there are two groups of them whose supports intersect at one coordinate, or one of them is $[r + \delta - 2, \delta - 1, r]$ MDS code, and $s-1$ of them are $[r + \delta - 1, \delta - 1, r+1]$ MDS codes with disjoint supports.
From Proposition \ref{MDS4}, $n' \leq r + \delta - 1 \leq 6$, $r + \delta - 1=6$.

The code $\langle H_{j}' \rangle $ corresponding to the remaining group $H_j$ is an $[\#S_j, \delta - 1, \#S_j-\delta + 2]$ MDS code with length $\#S_j \geq n'= 6$, then $\#S_j = 6$, $\langle H_{j}' \rangle = \langle H'\rangle$ is a quaternary $[6,3,4]$ MDS code.

In this case, the optimal $(r, \delta)$-LRC code has parameters
\begin{equation}\label{5}
n = 6l-1,~ k = 3l-1,~ d=4,~ r = 3, ~\delta = 4,~(l \geq 2).
\end{equation}

\begin{construction}\label{construction5}

 The following parity-check matrices $H$ gives the optimal quaternary $(r, \delta)$-LRC code for the parameters in (\ref{5}).
$$
H=\left(\begin{array}{c|c}
\tilde{H} & \mathbf{0}_{6 \times(6(l-2))} \\
\hline \mathbf{0}_{(3(l-2)) \times 11} & I_{l-2} \otimes\left(\begin{array}{llllll}
1 & 0 & 0 & 1 & 1 & 1\\
0 & 1 & 0 & 1 & w & w^2\\
0 & 0 & 1 & 1  & w^2 & w
\end{array}\right)
\end{array}\right),
$$
where
$$
\tilde{H}=\left(\begin{array}{ccccccccccc}
1 & 0 & 0 & 1 & 1 & \alpha   & 0  & 0 &  0 & 0  & 0\\
0 & 1 & 0 & 1 & w & \beta    & 0  & 0 &  0 & 0  & 0\\
0 & 0 & 1 & 1 & w^2 & \gamma & 0  & 0 &  0 & 0  & 0\\
0 & 0 & 0 & 0 & 0 & 1        & 0 & 0 & 1 & 1 & 1\\
0 & 0 & 0 & 0 & 0 & 0        & 1 & 0 & 1 & w & w^2\\
0 & 0 & 0 & 0 & 0 & 0        & 0 & 1 & 1  & w^2 & w
\end{array}\right),
$$
and $(\alpha,\beta,\gamma)^{T}=(0,0,0)^{T}$ or $(1, w^2,w)^{T}$.

\end{construction}

 \subsubsection{If $n'= 5$.}\label{4.2.2}
 $\mathcal C^\prime$ is a quaternary $[5,2,4]$ MDS code, $\langle H'\rangle$ is a quaternary $[5,3,3]$ MDS code.
$$n-n'= (\lceil\frac{k}{r}\rceil-1)(r + \delta - 1),$$
that is, the supports of the $s$ removed groups are disjoint and each has support size exactly $r + \delta - 1$.
These $s$ codes $\langle H_{i}' \rangle $ corresponding to the removed groups $H_i$ are all $[r + \delta - 1, \delta - 1, r+1]$ MDS codes, similar to above, we have $r + \delta - 1=5 $ or $6$.

\underline{\textbf{Case $r + \delta - 1=5 $.}}

By $\delta >2 $ and $t=2 <r$, we have $[r + \delta - 1, \delta - 1, r+1]=[5,2,4]$.
The code $\langle H_{j}' \rangle $ corresponding to the remaining group $H_j$ is an $[\#S_j, \delta - 1, \#S_j-\delta + 2]$ MDS code with length $n' \leq \#S_j \leq r + \delta - 1$, then $\#S_j = 5$, $\langle H_{l}' \rangle = \langle H'\rangle$ is a quaternary $[5,3,3]$ MDS code.

In this subcase, the optimal $(r, \delta)$-LRC code has parameters
\begin{equation}\label{6}
n = 5l,~ k = 3l-1,~ d=4,~ r = 3, ~\delta = 3,~(l \geq 2).
\end{equation}


\begin{construction}\label{construction6}

The following parity-check matrix $H$ gives the optimal quaternary $(r, \delta)$-LRC code for the parameters in (\ref{6}).

$$
H=\left(\begin{array}{c}
I_{l} \otimes\left(\begin{array}{lllll}
1 & 0 & 1 & 1 & 1\\
0 & 1 & 1 & w & w^2\\
\end{array}\right)\\

\hline

\textbf{1}_l \otimes\left(\begin{array}{lllll}
0 & 0 & 1 & w^2 & w
\end{array}\right)
\end{array}\right).
$$
\end{construction}

\underline{\textbf{Case $r + \delta - 1=6 $.}}

In this subcase, we have $[r + \delta - 1, \delta - 1, r+1]=[6,3,4]$,
and the code $\langle H_{j}' \rangle $ corresponding to the remaining group $H_j$ is an $[\#S_j, \delta - 1, \#S_j-\delta + 2]$ MDS code with length $\#S_j = 5$ or $6$, $\langle H_{l}' \rangle$ is a quaternary $[5,3,3]$ or $[6,3,4]$ MDS code.

The optimal quaternary $(r, \delta)$-LRC codes in this subcase have the same structure as the optimal quaternary $(r, \delta)$-LRC codes in case \ref{4.2.1} with $n'=6$.

\subsection{Case $t=1$.}\label{4.3}

When $t=1$, the length of $H'$ is
$$4 \leq n' \leq 6.$$
By $r-1 \geq t$, we have $r \geq 2.$

\subsubsection{Case $r=2$.}\label{4.3.1}
We know that $\delta =3$ or 4.

\underline{\textbf{If $\delta =3$.}}
By $n' \leq r+ \delta -1$, we have $n'=r+ \delta -1=4$. The code $\mathcal C^\prime$ is a quaternary $[4,1,4]$ MDS code, $\langle H'\rangle$ is a quaternary $[4,3,2]$ MDS code.
$$n-n'= (\lceil\frac{k}{r}\rceil-1)(r + \delta - 1), $$
the supports of the $\lceil\frac{k}{r}\rceil-1=s$ removed groups are disjoint and each has support size exactly $r + \delta - 1$.
The $s$ codes $\langle H_{i}' \rangle $ corresponding to the removed groups $H_i$ are all $[r + \delta - 1, \delta - 1, r+1]=[4,2,3]$ MDS codes.
And the code $\langle H_{j}' \rangle $ corresponding to the remaining group $H_j$ is also a $[4,2,3]$ MDS code.

In this subcase, the optimal $(r, \delta)$-LRC code has parameters
\begin{equation}\label{7}
n = 4l,~ k = 2l-1,~ d=4,~ r = 2, ~\delta = 3,~(l \geq 2).
\end{equation}

\begin{construction}\label{construction7}

By puncturing exactly one column from each local group of the matrix in Construction \ref{construction6}, we can get the parity-check matrix $H$ of the optimal quaternary $(r, \delta)$-LRC code for the parameters in (\ref{7}), as follows.

$$
H=\left(\begin{array}{c}
I_{l} \otimes\left(\begin{array}{llll}
1 & 0 & 1 & 1 \\
0 & 1 & 1 & w \\
\end{array}\right)\\

\hline

\textbf{1}_l \otimes\left(\begin{array}{llll}
0 & 0 & 1 & w^2
\end{array}\right)
\end{array}\right).
$$

\end{construction}

\underline{\textbf{If $\delta =4$.}}
By $n' \leq r+ \delta -1=5$, we have $n'=4$ or $5$. The code $\mathcal C^\prime$ is a quaternary $[4,1,4]$ or $[5, 2, 4] $ MDS code, accordingly, $\langle H'\rangle$ is a quaternary $[4,3,2]$ or $[5, 3, 3]$ MDS code.

Similar to case \ref{4.2} with $r+\delta -1=6$, we can get the codes $\langle H_{i}' \rangle$, $i\in [l]$, are either all $[r + \delta - 1, \delta - 1, r+1]=[5,3,3]$ MDS codes and there are two groups in which their supports intersect at one coordinate, or one of them is $[4, 3, 2]$ MDS code, and $l-1$ of them are $[5, 3, 3]$ MDS codes with disjoint supports.

In this subcase, the optimal $(r, \delta)$-LRC code has parameters
\begin{equation}\label{8}
n = 5l-1,~ k = 2l-1,~ d=4,~ r = 2, ~\delta = 4,~(l \geq 2).
\end{equation}


\begin{construction}\label{construction8}

 By puncturing exactly one column from each local group of the matrix in Construction \ref{construction5}, we can get the parity-check matrix $H$ of the optimal quaternary $(r, \delta)$-LRC code for the parameters in (\ref{8}), as follows.

$$
H=\left(\begin{array}{c|c}
\tilde{H} & \mathbf{0}_{6 \times(5(l-2))} \\
\hline \mathbf{0}_{(3(l-2)) \times 9} & I_{l-2} \otimes\left(\begin{array}{lllll}
1 & 0 & 0 & 1 & 1 \\
0 & 1 & 0 & 1 & w^2 \\
0 & 0 & 1 & 1 & w
\end{array}\right)
\end{array}\right),
$$

where
$$
\tilde{H}=\left(\begin{array}{ccccccccccc}
1 & 0 & 0 & 1  & \alpha   & 0 & 0 & 0 & 0\\
0 & 1 & 0 & 1  & \beta    & 0 & 0 & 0 & 0  \\
0 & 0 & 1 & 1  & \gamma   & 0 & 0 & 0 & 0  \\
0 & 0 & 0 & 0  & 1        & 0 & 0 & 1 &  1\\
0 & 0 & 0 & 0  & 0        & 1 & 0 & 1 &  w^2\\
0 & 0 & 0 & 0  & 0        & 0 & 1 & 1 &  w
\end{array}\right),
$$
and $(\alpha,\beta,\gamma)^{T}=(0,0,0)^{T}$ or $(1, w^2,w)^{T}$.
\end{construction}

\subsubsection{Case $r=3$.}\label{4.3.2}
~

\underline{\textbf{If $\delta =3$.}}
By $n' \leq r+ \delta -1=5$, we have $n'=4$ or $5$. Th code $\mathcal C^\prime$ is  then a quaternary $[4,1,4]$ or $[5,2,4]$ MDS code, accordingly, $\langle H'\rangle$ is a quaternary $[4,3,2]$ or $[5,3,3]$ MDS code.

The codes $\langle H_{i}' \rangle$, $i\in[l]$, are either all $[r + \delta - 1, \delta - 1, r+1]=[5,2,4]$ MDS code and there are two groups in which their supports intersect at one coordinate, or one of them is $[4, 2, 3]$ MDS code, and $l-1$ of them are $[5, 2, 4]$ MDS codes with disjoint supports.

In this subcase, the optimal $(r, \delta)$-LRC code has parameters
\begin{equation}\label{9}
n = 5l-1,~ k = 3l-2,~ d=4,~ r = 3, ~\delta = 3,~(l \geq 2).
\end{equation}


\begin{construction}\label{construction9}

 The following parity-check matrix $H$ gives the optimal quaternary $(r, \delta)$-LRC code for the parameters in (\ref{9}).

$$
H=\left(\begin{array}{c|c}
\begin{array}{ccccccccc}
1 & 0 & 1 & 1 & \alpha & 0  & 0 & 0 & 0\\
0 & 1 & 1 & w & \beta  & 0  & 0 & 0 & 0\\
0 & 0 & 0 & 0 & 1      & 0  & 1 & 1 & 1\\
0 & 0 & 0 & 0 & 0      & 1  & 1 & w & w^2
\end{array} & \mathbf{0}_{4 \times(5(l-2))} \\
\hline \mathbf{0}_{(2(l-2)) \times 9} & I_{l-2} \otimes\left(\begin{array}{lllll}
1 & 0 & 1 & 1 & 1\\
0 & 1 & 1 & w & w^2
\end{array}\right)\\
\hline
\begin{array}{lllllllll}
0 & 0 & 1 & w^2 & 0 & 0 & 1 & w^2 & w
\end{array} & \textbf{1}_l \otimes\left(\begin{array}{lllll}
0 & 0 & 1 & w^2 & w
\end{array}\right)
\end{array}\right).
$$

\end{construction}

\underline{\textbf{If $\delta =4$.}}
By $n' \leq r+ \delta -1=6$, we have $n'=4, 5 ~\mbox{or} ~6$.

In this subcase, the optimal $(r, \delta)$-LRC code has parameters
\begin{equation}\label{10}
n = 6l -2,~ k = 3l-2,~ d=4,~ r = 3, ~\delta = 4,~(l \geq 2).
\end{equation}

Comparing the parameters in (\ref{10}) with the parameters in (\ref{5}), it is not difficult  to get the following construction.


\begin{construction}\label{construction10}

  By puncturing the first column from the matrix $H$ in Construction \ref{construction5}, we can get the parity-check matrix of the optimal quaternary $(r, \delta)$-LRC code for the parameters in (\ref{10}).

\end{construction}

\subsubsection{Case $r\geq 4$.}\label{4.3.3}
~

In this case, if $\delta =3$, $r+ \delta -1 \geq 6$. From Corollary \ref{delta}, there exists  $i\in [l]$, such that the code $\langle H_{i}' \rangle $ is a quaternary $[\tilde{n}\geq 6, 2, \tilde{n}+1]$ MDS, this contradicts Proposition \ref{MDS4}. If $\delta =4$, $r+ \delta -1 \geq 7$. From Corollary \ref{delta}, there exists $i\in [l]$, such that the code $\langle H_{i}' \rangle $ is a quaternary $[\tilde{n}\geq 7, 3, \tilde{n}+1]$ MDS, this also contradicts Proposition \ref{MDS4}. Therefore, there is no optimal $(r, \delta)$-LRC code in this case.

\subsection{Case $t=0, r\leq3$.}\label{4.3.4}

When $t=0$, we have $r | k$. From Lemma \ref{k=rs}, the codes $\langle H_{i}' \rangle$, $i\in[l]$, are all $[r+ \delta -1, \delta -1, r+1]$ MDS code with disjoint supports and $n=l(r+ \delta -1)$.
By $k=rs$, $n = k+d-1+(\lceil\frac{k}{r}\rceil-1)(\delta-1)=s(r+ \delta -1)+4-\delta$. From $l=\lceil\frac{k}{r}\rceil=s$, we have $\delta =4$.

Hence, the optimal $(r, \delta)$-LRC code has parameters
\begin{equation}\label{11}
n = l(r+3),~ k = rl,~ d=4,~ 1 \leq r \leq 3, ~\delta = 4,~(l \geq 2).
\end{equation}


\begin{construction}\label{construction11}

For the parameters in (\ref{11}), the following parity-check matrix $H$ gives the optimal quaternary $(r, \delta)$-LRC code with $r=3$.
By puncturing precisely one or two columns from each local group of the matrix $H$, we can obtain the parity-check matrix of the optimal quaternary $(r, \delta)$-LRC codes with $r=2$ or $1$, respectively.

$$
H=\left(\begin{array}{c}
 I_{l} \otimes\left(\begin{array}{llllll}
1 & 0 & 0 & 1 & 1 & 1\\
0 & 1 & 0 & 1 & w & w^2\\
0 & 0 & 1 & 1 & w^2 & w
\end{array}\right)
\end{array}\right).
$$
\end{construction}

\section{Analyzing optimal quaternary $(r, \delta)$-LRC codes with minimum distance $d\geq5$}\label{Sec:LRC-$d=5$}

From Proposition \ref{H'MDS}(3), $\mathcal C^\prime$ is a quaternary $[n',1,n']$ MDS code, $\langle H'\rangle$ is a quaternary $[n',n'-1,2]$ MDS code, where $n'=d$.

By Lemma \ref{distance}, since $d > q =4$, we have $5 \leq d\leq \delta q= 4\delta$ and we can also get $r \mid(k-1)$, that is,
$$r=1 ~\mbox{or}~ r\geq 2,~ k=rs+1.$$

\subsection{Case $r=1$.}

In this case, we have $r \mid k$, from Lemma \ref{k=rs}, $n=l(r+ \delta -1)=l\delta$,
the optimal $(r, \delta)$-LRC code has parity-check matrix $H$ of the following form
$$
H=\left[\begin{array}{ccc|ccc}
I_{\delta-1} &       &             & Q_{1}  &        & \\
             &\ddots &             &        & \ddots & \\
             &       & I_{\delta-1}&        &        &Q_{l} \\
\hline
   &0_{(l-k)\times l(\delta-1)}&   & A_{1}  &\cdots  & A_{l}
\end{array}\right],
$$
where the matrices $Q_i$ and $A_i$ appearing above are of sizes $(\delta - 1) \times 1$ and $(l-k)\times 1$, respectively.

The matrices $ H_{i}'  = (I_{\delta-1}~Q_i)$, $i\in[l]$, are all the parity-check matrices of $[\delta, 1, \delta]$ MDS codes, there is a linear combination of the columns of $H_{i}'$ equals to $0$. This implies that there is a linear combination of the columns of $H$ corresponding to $H_{i}'$ equals to $(0, \cdots,0, A_i^T)^T$.

Since
$d=n-k+1-(\lceil\frac{k}{r}\rceil-1)(\delta-1)=\delta(l-k+1)$,
any $l-k$ columns of matrix $A=(A_1~\cdots ~ A_l)$ need to be linearly independent.
Meanwhile, from $0 < \delta(l-k+1) \leq 4\delta$, we have $k \leq l \leq k+3$.

\subsubsection{If $l=k$.}

The optimal $(r, \delta)$-LRC code has parameters
\begin{equation}\label{12}
n = k\delta ,~ k \geq 2,~ d=\delta,~ r = 1, ~\delta \geq 5.
\end{equation}

Because $n-k-l(\delta-1)=l-k=0$, there is no global row in parity-check matrix $H$.


\begin{construction}\label{construction12}

Let $\tilde{H}$ be the generator matrix of the quaternary $[\delta, \delta-1, 2]$ $(\delta \geq 5)$ MDS code. Then the following parity-check matrix $H$ gives the optimal quaternary $(r, \delta)$-LRC code for the parameters in (\ref{12}).
$$
H=\left(\begin{array}{c}
 I_{k} \otimes \tilde{H}
\end{array}\right).
$$

\end{construction}

\subsubsection{If $l=k+1$.}

The optimal $(r, \delta)$-LRC code has parameters
\begin{equation}\label{13}
n = (k+1)\delta ,~ k \geq 2,~ d=2\delta,~ r = 1, ~\delta >2.
\end{equation}

There is exactly $l-k=1$ global row in parity-check matrix $H$.


\begin{construction}\label{construction13}

Let $\tilde{H}$ be the generator matrix of the quaternary $[\delta, \delta-1, 2]$ $(\delta > 2)$ MDS code. Then the following parity-check matrix $H$ gives the optimal quaternary $(r, \delta)$-LRC code for the parameters in (\ref{13}).
$$
H=\left(\begin{array}{c}
 I_{k+1} \otimes \tilde{H}\\
 \hline
 \textbf{1}_{k+1} \otimes \left(\begin{array}{llll}
0 & ... & 0 & 1
\end{array}\right)_{1\times \delta}
\end{array}\right).
$$

\end{construction}

\subsubsection{If $l=k+2$.}

Since any $l-k=2$ columns of matrix $A$ need to be linearly independent, we have $l \leq \frac{q^2-1}{q-1}=5$.

The optimal $(r, \delta)$-LRC code has parameters
\begin{equation}\label{14}
n = (k+2)\delta ,~ 2\leq k \leq  3,~ d=3\delta,~ r = 1, ~\delta >2.
\end{equation}

There are exactly two global rows in parity-check matrix $H$.


\begin{construction}\label{construction14}

Let $\tilde{H}$ be the generator matrix of the quaternary $[\delta, \delta-1, 2]$ $(\delta > 2)$ MDS code. Then the following parity-check matrix $H$ gives the optimal quaternary $(r, \delta)$-LRC code for the parameters in (\ref{14}) with $k=3$. By puncturing the rows and columns corresponding to a  local group of the matrix $H$, we obtain the parity-check matrix of the optimal quaternary $(r, \delta)$-LRC code with $k=2$.

$$
H=\left(\begin{array}{c}
I_{5} \otimes \tilde{H} \\
\hline
A_{(1~0)}~~ A_{(0~1)}~~ A_{(1~1)} ~~A_{(1~w)} ~~ A_{(1~w^2)}
\end{array}\right),
$$
where $A_{(\alpha~\beta)}=\left(\begin{array}{llll}
0 & ... & 0 & \alpha\\
0 & ... & 0 & \beta
\end{array}\right)_{2 \times \delta}.$

\end{construction}

\subsubsection{If $l=k+3$.}

Since any $l-k=3$ columns of matrix $A$ has to be linearly independent, we have $l \leq 6$.

The optimal $(r, \delta)$-LRC code has parameters
\begin{equation}\label{15}
n = (k+3)\delta ,~ 2\leq k \leq  3,~ d=4\delta,~ r = 1, ~\delta >2.
\end{equation}

There are exactly three global rows in parity-check matrix $H$.


\begin{construction}\label{construction15}

Let $\tilde{H}$ be the generator matrix of the quaternary $[\delta, \delta-1, 2]$ $(\delta > 2)$ MDS code. Then the following parity-check matrix $H$ gives the optimal quaternary $(r, \delta)$-LRC code for the parameters in (\ref{15}) with $k=3$. By puncturing the rows and columns corresponding to a local group of the matrix $H$, we can obtain the parity-check matrix of the optimal quaternary $(r, \delta)$-LRC code with $k=2$.

$$
H=\left(\begin{array}{c}
I_{6} \otimes \tilde{H} \\
\hline
A_{(1~0~0)} ~~ A_{(0~1~0)} ~~ A_{(0~0~1)} ~~ A_{(1~1~1)} ~~ A_{(1~w~w^2)} ~~ A_{(1~w^2~w)}
\end{array}\right),
$$
where $A_{(\alpha~\beta~\gamma)}=\left(\begin{array}{llll}
0 & ... & 0 & \alpha\\
0 & ... & 0 & \beta\\
0 & ... & 0 & \gamma\\
\end{array}\right)_{3 \times \delta}.$


\end{construction}

\subsection{Case $k=rs+1,~r \geq 2$.}

In this case, $n=k+d-1+(\lceil\frac{k}{r}\rceil-1)(\delta-1)= s(r+\delta-1)+d$.
 $$n-n'=n-d=s(r+\delta-1),$$ 
this means that the $\lceil\frac{k}{r}\rceil-1=s$ removed local groups have uniform support size $\#S_j=r + \delta - 1$. Since these $s$ groups can be chosen arbitrarily, all local groups $H_i$, $i\in[l]$, have supports size $\#S_i=r + \delta - 1$, and the code $\langle H_{i}' \rangle $ is an $[r + \delta - 1, \delta - 1, r+1]$ MDS code.
By Proposition \ref{MDS4}, since $r \geq 2$, $\delta > 2$, we have $[r + \delta - 1, \delta - 1, r+1]=[4,2,3],~[5,3,3],~[5,2,4] ~\mbox{or}~ [6,3,4]$.

\subsubsection{Case $r=2,~\delta=3$.}\label{5.2.1}

We have $k=2s+1,$ the codes $\langle H_{i}' \rangle $, $i\in[l]$, are $[4,2,3]$ MDS codes. The codes $\mathcal C|_{S_i}$, $i\in[l]$, are also $[4,2,3]$ MDS codes.

\underline{\textbf{If $s=1$.}} $n=4+d$, $k = 3$ and $5 \leq d\leq  4\delta=12$.

When $d=12$, the optimal LRC code $\mathcal C$ is a $[16,3,12]$ linear code with locality $(2, 3)$. We can think of the columns of generator matrix $G$ as points in the projective plane $PG(2,\F_4)$. There are $\frac{4^3-1}{4-1}=21$ points in $PG(2,\F_4)$, each line in $PG(2,\F_4)$ contains five points and any two lines intersect precisely at one point. Fix a point $\textbf{a}$, there are five lines passing through this point, we know that these five lines cover all $21$ points in $PG(2,\F_4)$, as shown in Figure 1.

\begin{tikzpicture}

\centering

\node (f1) at (1,-4)  {\textbf{Figure 1} };

\draw  [black](-1,0)--(3,0);
\filldraw [gray] (-1,0) circle (3pt);
\filldraw [gray] (0,0) circle (1pt);
\filldraw [gray] (1,0) circle (1pt);
\filldraw [gray] (2,0) circle (1pt);
\filldraw [gray] (3,0) circle (1pt);
\node (A) at (-1,0)  {$a$};
\node (31) at (0,0)  {$a_{31}$};
\node (32) at (1,0)  {$a_{32}$};
\node (33) at (2,0)  {$a_{33}$};
\node (34) at (3,0)  {$a_{34}$};
\node (3) at (3.8,0)  {\textbf{$l_3$} };

\draw  [black](-1,0)--(2.6,1.8);
\filldraw [gray] (-0.1,0.45) circle (1pt);
\filldraw [gray] (0.8,0.9) circle (1pt);
\filldraw [gray] (1.7,1.35) circle (1pt);
\filldraw [gray] (2.6,1.8) circle (1pt);
\node (21) at (-0.1,0.45)  {$a_{21}$};
\node (22) at (0.8,0.9)  {$a_{22}$};
\node (23) at (1.7,1.35)  {$a_{23}$};
\node (24) at (2.6,1.8)  {$a_{24}$};
\node (2) at (3.4,1.8)  {\textbf{$l_2$} };

\draw  [black](-1,0)--(1.8,3.36);
\filldraw [gray] (-0.3,0.84) circle (1pt);
\filldraw [gray] (0.4,1.68) circle (1pt);
\filldraw [gray] (1.1,2.52) circle (1pt);
\filldraw [gray] (1.8,3.36) circle (1pt);
\node (11) at (-0.3,0.84)  {$a_{11}$};
\node (12) at (0.4,1.68)  {$a_{12}$};
\node (13) at (1.1,2.52)  {$a_{13}$};
\node (14) at (1.8,3.36)  {$a_{14}$};
\node (1) at (2.6,3.36)  {\textbf{$l_1$} };

\draw  [black](-1,0)--(2.6,-1.8);
\filldraw [gray] (-0.1,-0.45) circle (1pt);
\filldraw [gray] (0.8,-0.9) circle (1pt);
\filldraw [gray] (1.7,-1.35) circle (1pt);
\filldraw [gray] (2.6,-1.8) circle (1pt);
\node (41) at (-0.1,-0.45)  {$a_{41}$};
\node (42) at (0.8,-0.9)  {$a_{42}$};
\node (43) at (1.7,-1.35)  {$a_{43}$};
\node (44) at (2.6,-1.8)  {$a_{44}$};
\node (4) at (3.4,-1.8)  {\textbf{$l_4$} };

\draw  [black](-1,0)--(1.8,-3.36);
\filldraw [gray] (-0.3,-0.84) circle (1pt);
\filldraw [gray] (0.4,-1.68) circle (1pt);
\filldraw [gray] (1.1,-2.52) circle (1pt);
\filldraw [gray] (1.8,-3.36) circle (1pt);
\node (51) at (-0.3,-0.84)  {$a_{51}$};
\node (52) at (0.4,-1.68)  {$a_{52}$};
\node (53) at (1.1,-2.52)  {$a_{53}$};
\node (54) at (1.8,-3.36)  {$a_{54}$};
\node (5) at (2.6,-3.36)  {\textbf{$l_5$} };

\node (f2) at (7,-4)  {\textbf{Figure 2} };

\draw  [black, densely dashed](5,0)--(6,0);
\draw  [black](6,0)--(9,0);
\filldraw [gray] (5,0) circle (3pt);
\filldraw [black] (6,0) circle (1pt);
\filldraw [black] (7,0) circle (1pt);
\filldraw [black] (8,0) circle (1pt);
\filldraw [black] (9,0) circle (1pt);
\node (31) at (6,0)  {$a_{31}$};
\node (32) at (7,0)  {$a_{32}$};
\node (33) at (8,0)  {$a_{33}$};
\node (34) at (9,0)  {$a_{34}$};
\node (3) at (9.8,0)  {\textbf{$l_3$} };

\draw  [black, densely dashed](5,0)--(5.9,0.45);
\draw  [black](5.9,0.45)--(8.6,1.8);
\filldraw [black] (5.9,0.45) circle (1pt);
\filldraw [black] (6.8,0.9) circle (1pt);
\filldraw [black] (7.7,1.35) circle (1pt);
\filldraw [black] (8.6,1.8) circle (1pt);
\node (21) at (5.9,0.45)  {$a_{21}$};
\node (22) at (6.8,0.9)  {$a_{22}$};
\node (23) at (7.7,1.35)  {$a_{23}$};
\node (24) at (8.6,1.8)  {$a_{24}$};
\node (2) at (9.4,1.8)  {\textbf{$l_2$} };

\draw  [black, densely dashed](5,0)--(5.7,0.84);
\draw  [black](5.7,0.84)--(7.8,3.36);
\filldraw [black] (5.7,0.84) circle (1pt);
\filldraw [black] (6.4,1.68) circle (1pt);
\filldraw [black] (7.1,2.52) circle (1pt);
\filldraw [black] (7.8,3.36) circle (1pt);
\node (11) at (5.7,0.84)  {$a_{11}$};
\node (12) at (6.4,1.68)  {$a_{12}$};
\node (13) at (7.1,2.52)  {$a_{13}$};
\node (14) at (7.8,3.36)  {$a_{14}$};
\node (1) at (8.6,3.36)  {\textbf{$l_1$} };

\draw  [black, densely dashed](5,0)--(5.9,-0.45);
\draw  [black](5.9,-0.45)--(8.6,-1.8);
\filldraw [black] (5.9,-0.45) circle (1pt);
\filldraw [black] (6.8,-0.9) circle (1pt);
\filldraw [black] (7.7,-1.35) circle (1pt);
\filldraw [black] (8.6,-1.8) circle (1pt);
\node (41) at (5.9,-0.45)  {$a_{41}$};
\node (42) at (6.8,-0.9)  {$a_{42}$};
\node (43) at (7.7,-1.35)  {$a_{43}$};
\node (44) at (8.6,-1.8)  {$a_{44}$};
\node (4) at (9.4,-1.8)  {\textbf{$l_4$} };

\draw  [black, densely dashed](5,0)--(7.8,-3.36);
\filldraw [gray] (5.7,-0.84) circle (1pt);
\filldraw [gray] (6.4,-1.68) circle (1pt);
\filldraw [gray] (7.1,-2.52) circle (1pt);
\filldraw [gray] (7.8,-3.36) circle (1pt);
\node (5) at (8.6,-3.36)  {\textbf{$l_5$} };

\end{tikzpicture}

We  need to find $16$ points from $PG(2,\F_4)$ to form matrix $G$.

 To ensure a locality equals $(r=2, \delta=3)$ for the code $\mathcal C$, these points need to be such that each point is on a line segment consisting of $r+\delta-1=4$ points. Besides, to get a minimum distance equal $d=12$, these points need to be such that any $n-d+1=5$ points are not on the same line. Delete any one line in Figure 1, leaving four parallel line segments in $PG(2,\F_4)$ with a total of 16 points, as shown in Figure 2. Using these 16 points as columns, we can obtain the desired matrix $G$.

When $d=11$, the optimal LRC code $\mathcal C$ is a $[15,3,11]$ linear code with locality $(2, 3)$. We want to find $15$ points from $PG(2,\F_4)$ to form matrix $G$. Similarly, to ensure the locality and the minimum distance  of code $\mathcal C$, these points need to be such that each point is on a line segment of length $r+\delta-1=4$ and any other $n-d+1=5$ points are not on the same line.
Choose any point $\textbf{b}$ from the broken line $l_5$ in Figure 2, and there are also five lines passing through this point in $PG(2,\F_4)$. Suppose $l_6$ is a line passing through point $\textbf{b}$ and $l_6$ is not the broken line $l_5$,
then there is an intersection point between $l_6$ and each of the four lines $l_i$, $i\in[4]$, denoted as $\textbf{b}_{6i}$, respectively.
Let $l_7$ is another line passing through point $\textbf{b}$, and $l_7$ is neither the broken line $l_5$ nor the line $l_6$, then there is also an intersection point between $l_7$ and each of the four lines $l_i$, $i\in[4]$, denoted as $\textbf{b}_{7i}$, respectively. Meanwhile $\{\textbf{b}_{61},\textbf{b}_{62},\textbf{b}_{63},\textbf{b}_{64}\} \cap \{\textbf{b}_{71},\textbf{b}_{72},\textbf{b}_{73},\textbf{b}_{74}\}= \emptyset$. In this way, we can find four lines passing through point $\textbf{b}$, as shown in Figure 3.

As shown in Figure 4, by deleting point $\textbf{b}$ and anyone point $\textbf{a}_{1i}$, $i\in [4]$, in Figure 3, using the remaining $15$ points as columns, we can get the desired matrix $G$. Further, according to the same principle, we can obtain the generation matrix of code $\mathcal C$ of all possible parameters in this case by selecting appropriate points.

\begin{tikzpicture}

\node (f3) at (1,-4)  {\textbf{Figure 3} };

\draw  [black, densely dashed](-1,0)--(0,0);
\draw  [black](0,0)--(3,0);
\filldraw [gray] (-1,0) circle (3pt);
\filldraw [gray] (0,0) circle (1pt);
\filldraw [gray] (1,0) circle (1pt);
\filldraw [gray] (2,0) circle (1pt);
\filldraw [gray] (3,0) circle (1pt);
\node (31) at (0,0)  {$a_{31}$};
\node (32) at (1,0)  {$a_{32}$};
\node (33) at (2,0)  {$a_{33}$};
\node (34) at (3,0)  {$a_{34}$};
\node (3) at (3.8,0)  {\textbf{$l_3$} };

\draw  [black, densely dashed](-1,0)--(-0.1,0.45);
\draw  [black](-0.1,0.45)--(2.6,1.8);
\filldraw [gray] (-0.1,0.45) circle (1pt);
\filldraw [gray] (0.8,0.9) circle (1pt);
\filldraw [gray] (1.7,1.35) circle (1pt);
\filldraw [gray] (2.6,1.8) circle (1pt);
\node (21) at (-0.1,0.45)  {$a_{21}$};
\node (22) at (0.8,0.9)  {$a_{22}$};
\node (23) at (1.7,1.35)  {$a_{23}$};
\node (24) at (2.6,1.8)  {$a_{24}$};
\node (2) at (3.4,1.8)  {\textbf{$l_2$} };

\draw  [black, densely dashed](-1,0)--(-0.3,0.84);
\draw  [black](-0.3,0.84)--(1.8,3.36);
\filldraw [gray] (-0.3,0.84) circle (1pt);
\filldraw [gray] (0.4,1.68) circle (1pt);
\filldraw [gray] (1.1,2.52) circle (1pt);
\filldraw [gray] (1.8,3.36) circle (1pt);
\node (11) at (-0.3,0.84)  {$a_{11}$};
\node (12) at (0.4,1.68)  {$a_{12}$};
\node (13) at (1.1,2.52)  {$a_{13}$};
\node (14) at (1.8,3.36)  {$a_{14}$};
\node (1) at (2.6,3.36)  {\textbf{$l_1$} };

\draw  [black, densely dashed](-1,0)--(-0.1,-0.45);
\draw  [black](-0.1,-0.45)--(2.6,-1.8);
\filldraw [gray] (-0.1,-0.45) circle (1pt);
\filldraw [gray] (0.8,-0.9) circle (1pt);
\filldraw [gray] (1.7,-1.35) circle (1pt);
\filldraw [gray] (2.6,-1.8) circle (1pt);
\node (41) at (-0.1,-0.45)  {$a_{41}$};
\node (42) at (0.8,-0.9)  {$a_{42}$};
\node (43) at (1.7,-1.35)  {$a_{43}$};
\node (44) at (2.6,-1.8)  {$a_{44}$};
\node (4) at (3.4,-1.8)  {\textbf{$l_4$} };

\draw  [black, densely dashed](-1,0)--(1.8,-3.36);
\filldraw [gray] (-0.3,-0.84) circle (1pt);
\filldraw [gray] (0.4,-1.68) circle (1pt);
\filldraw [gray] (1.1,-2.52) circle (1pt);
\filldraw [gray] (1.8,-3.36) circle (1pt);
\node (53) at (1.1,-2.52)  {$b$};
\node (5) at (2.6,-3.36)  {\textbf{$l_5$} };

\draw  [red](1.1,-2.52)--(-0.1,-0.45)--(0,0)--(-0.1,0.45)--(-0.3,0.84);
\draw  [red](1.1,-2.52)--(0.8,-0.9)--(1,0)--(0.8,0.9)--(0.4,1.68);
\draw  [red](1.1,-2.52)--(1.7,-1.35)--(2,0)--(1.7,1.35)--(1.1,2.52);
\draw  [red](1.1,-2.52)--(2.6,-1.8)--(3,0)--(2.6,1.8)--(1.8,3.36);

\node (f4) at (7,-4)  {\textbf{Figure 4} };

\draw  [black, densely dashed](5,0)--(6,0);
\draw  [black](6,0)--(9,0);
\filldraw [gray] (5,0) circle (3pt);
\filldraw [black] (6,0) circle (1pt);
\filldraw [black] (7,0) circle (1pt);
\filldraw [black] (8,0) circle (1pt);
\filldraw [black] (9,0) circle (1pt);
\node (31) at (6,0)  {$a_{31}$};
\node (32) at (7,0)  {$a_{32}$};
\node (33) at (8,0)  {$a_{33}$};
\node (34) at (9,0)  {$a_{34}$};
\node (3) at (9.8,0)  {\textbf{$l_3$} };

\draw  [black, densely dashed](6,0)--(6.9,0.45);
\draw  [black](6.9,0.45)--(8.6,1.8);
\filldraw [black] (5.9,0.45) circle (1pt);
\filldraw [black] (6.8,0.9) circle (1pt);
\filldraw [black] (7.7,1.35) circle (1pt);
\filldraw [black] (8.6,1.8) circle (1pt);
\node (21) at (5.9,0.45)  {$a_{21}$};
\node (22) at (6.8,0.9)  {$a_{22}$};
\node (23) at (7.7,1.35)  {$a_{23}$};
\node (24) at (8.6,1.8)  {$a_{24}$};
\node (2) at (9.4,1.8)  {\textbf{$l_2$} };

\draw  [black, densely dashed](5,0)--(5.7,0.84);
\draw  [black](5.7,0.84)--(7.8,3.36);
\filldraw [black] (5.7,0.84) circle (1pt);
\filldraw [black] (6.4,1.68) circle (1pt);
\filldraw [black] (7.1,2.52) circle (1pt);
\filldraw [black] (7.8,3.36) circle (1pt);
\node (11) at (5.7,0.84)  {$a_{11}$};
\node (12) at (6.4,1.68)  {$a_{12}$};
\node (13) at (7.1,2.52)  {$a_{13}$};
\node (14) at (7.8,3.36)  {$a_{14}$};
\node (1) at (8.6,3.36)  {\textbf{$l_1$} };

\draw  [black, densely dashed](5,0)--(8.6,-1.8);
\filldraw [gray] (5.9,-0.45) circle (1pt);
\filldraw [black] (6.8,-0.9) circle (1pt);
\filldraw [black] (7.7,-1.35) circle (1pt);
\filldraw [black] (8.6,-1.8) circle (1pt);
\node (42) at (6.8,-0.9)  {$a_{42}$};
\node (43) at (7.7,-1.35)  {$a_{43}$};
\node (44) at (8.6,-1.8)  {$a_{44}$};
\node (4) at (9.4,-1.8)  {\textbf{$l_4$} };

\draw  [black, densely dashed](5,0)--(7.8,-3.36);
\filldraw [gray] (5.7,-0.84) circle (1pt);
\filldraw [gray] (6.4,-1.68) circle (1pt);
\filldraw [gray] (7.1,-2.52) circle (1pt);
\filldraw [gray] (7.8,-3.36) circle (1pt);
\node (5) at (8.6,-3.36)  {\textbf{$l_5$} };

\draw  [red, densely dashed](7.1,-2.52)--(5.9,-0.45)--(6,0)--(5.9,0.45)--(5.7,0.84);
\draw  [red, densely dashed](7.1,-2.52)--(6.8,-0.9);
\draw  [red](6.8,-0.9)--(7,0)--(6.8,0.9)--(6.4,1.68);
\draw  [red, densely dashed](7.1,-2.52)--(7.7,-1.35);
\draw  [red](7.7,-1.35)--(8,0)--(7.7,1.35)--(7.1,2.52);
\draw  [red, densely dashed](7.1,-2.52)--(8.6,-1.8);
\draw  [red](8.6,-1.8)--(9,0)--(8.6,1.8)--(7.8,3.36);

\end{tikzpicture}







In this case, the optimal $(r, \delta)$-LRC code has parameters
\begin{equation}\label{16}
n = d+4 ,~ k= 3,~ 5 \leq d\leq 12,~ r = 2, ~\delta =3.
\end{equation}

Let matrix
$$
G=\left(\begin{array}{lllllllllllllllllll}
0 &0 &0  &0 & &1  &1&1&1 & &1  &1&1&1 & &1  &1&1&w^2\\
0 &1 &1  &1 & &w^2&1&w&0 & &w^2&w&0&1 & &w^2&0&1&1\\
1 &w &w^2&1 & &0  &0&0&0 & &1  &1&1&1 & &w  &w&w&1
\end{array}\right)_{3\times 16}.
$$
In fact, the columns of the matrix $G$ (denoted by $\textbf{g}_j$, $j\in[16]$) can correspond in turn to points $\textbf{a}_{1i}$, $\textbf{a}_{2i}$, $\textbf{a}_{3i}$ and $\textbf{a}_{4i}$, $i\in[4]$, in Figure 3. Meanwhile, we can take the point $\textbf{a}=(0,1,0)^T$ in the figure.

\begin{construction}\label{construction16}

The above generator matrix $G$ gives the optimal quaternary $(r, \delta)$-LRC code $\mathcal C$ for the parameters in (\ref{16}) with $d=12$. By puncturing $\mathcal C$ on the coordinates set $S$, with $S=\{13\}$, $\{13,14\}$, $\{13,14,15\}$, $\{13,14,15,16\}$ or $\{11,12,13,14,15,16\}$, we can get the optimal quaternary $(r, \delta)$-LRC codes with $d=11,~10,~9$, $8$ or $6$, respectively.
The matrix $G'=(\textbf{a}, \textbf{g}_1, \textbf{g}_2, \textbf{g}_3, \textbf{g}_5, \textbf{g}_6, \textbf{g}_7, \textbf{g}_9, \textbf{g}_{10}, \textbf{g}_{11}, \textbf{g}_{13})$ or $(\textbf{a}, \textbf{g}_1, \textbf{g}_2, \textbf{g}_3, \textbf{g}_5, \textbf{g}_6, \textbf{g}_7, \textbf{g}_9,\textbf{g}_{13})$ give the optimal quaternary $(r, \delta)$-LRC codes with $d=7$ or $5$, respectively.

\end{construction}

\underline{\textbf{If $s\geq 2$.}} $n=4s+d$, $k=2s+1$ and $5 \leq d \leq 12$.
Since $n-n'=s(r+\delta-1),$ the supports of $s$ removed local groups are pairwise disjoint. These $s$ groups can be chosen arbitrarily, we have the supports of all $l$ local groups are disjoint and $n=l(r+\delta-1)=4l$.
From $4s+d=4l$, we get $4 \mid d$, $d=8$ or $12$, $l-s=\frac{d}{4}=2$ or $3$.

\underline{\textbf{When $d=8$, $l-s=2$.}}

In this case, the optimal $(r, \delta)$-LRC has an equivalent parity-check matrix $H$ in the following form:
$$
H=\left(\begin{array}{cccc|cccc|c|cccc}
1 & 0 & 1 & 1 & 0 & 0 & 0 & 0 & \ldots  & 0 & 0 & 0 & 0\\
0 & 1 & 1 & w & 0 & 0 & 0 & 0 & \ldots  & 0 & 0 & 0 & 0\\
0 & 0 & 0 & 0 & 1 & 0 & 1 & 1 & \ldots  & 0 & 0 & 0 & 0\\
0 & 0 & 0 & 0 & 0 & 1 & 1 & w & \ldots  & 0 & 0 & 0 & 0 \\
\vdots & \vdots & \vdots & \vdots & \vdots & \vdots & \vdots & \vdots & \ldots & \vdots & \vdots & \vdots & \vdots \\
0 & 0 & 0 & 0 & 0 & 0 & 0 & 0 & \ldots  &1 & 0 & 1 & 1 \\
0 & 0 & 0 & 0 & 0 & 0 & 0 & 0 & \ldots  &0 & 1 & 1 & w \\
\hline
\textbf{0} & \textbf{0} & \textbf{u}_{1} & \textbf{v}_{1} & \textbf{0} &\textbf{0} & \textbf{u}_{2}& \textbf{v}_{2}& \ldots & \textbf{0} &\textbf{0}& \textbf{u}_{l} &  \textbf{v}_{l}
\end{array}\right)
$$
where $\textbf{u}_{i},~ \textbf{v}_{i} \in \F_4^3$, $i \in [l]$.
Denote $\mathcal{V}_i = \mbox{Span}\{\textbf{u}_i, \textbf{v}_i\}$ the subspace of $\F_4^3$ spanned by $\textbf{u}_{i}$ and $ \textbf{v}_{i}$. Since the minimum distance $d=8$, it's easy to verify that these vectors satisfy the following two properties:
\begin{enumerate}
\item [(1)] dim$(\mathcal{V}_i)$=2;
\item [(2)] for any $j \neq i \in [l]$, $\textbf{u}_i$, $\textbf{v}_i$, $\textbf{u}_i + \textbf{v}_i$, $w\textbf{u}_i + \textbf{v}_i \notin \mathcal{V}_j$.

\end{enumerate}

We can consider the vectors $\textbf{u}_{i},~ \textbf{v}_{i}$, $i \in [l]$, as points in the projective plane $PG(2,\F_4)$, then $\mathcal{V}_i$ is a line in $PG(2,\F_4)$ generated by $\textbf{u}_{i}$ and $ \textbf{v}_{i}$. Any two lines intersect precisely at one point, for any $i \neq j \in [l]$, let $\textbf{p}_{ij}$ be the intersection point of $\mathcal{V}_i$ and $\mathcal{V}_j$. From above properties, we have $\textbf{p}_{ij}=w^2\textbf{u}_i + \textbf{v}_i=w^2\textbf{u}_j + \textbf{v}_j$. This means that all the lines $\mathcal{V}_i$, $i \in [l]$, intersect at the same point $\textbf{p}=w^2\textbf{u}_1 + \textbf{v}_1=\cdots=w^2\textbf{u}_l + \textbf{v}_l$. There are at most five lines passing through a given point, therefore $l \leq 5$. Since $l=s+2$, $s\geq 2$, we get $l=4$ or 5, accordingly, $n=16$ or 20.

In this case, the optimal $(r, \delta)$-LRC code has parameters
\begin{equation}\label{l-s=2_1}
n = 16 ~\mbox{or} ~20,~ k= 3,~ d=8,~ r = 2, ~\delta =3.
\end{equation}

\begin{construction}\label{constructionl-s=2_1}
The following parity-check matrix $H$ gives the optimal quaternary $(r, \delta)$-LRC code for the parameters in (\ref{l-s=2_1}) with $n=20$. By puncturing the rows and columns corresponding to a local group of the matrix $H$, we can obtain the parity-check matrix of the optimal quaternary $(r, \delta)$-LRC code with $n=16$.



$$
H=\left(\begin{array}{c}
I_{5} \otimes \tilde{H} \\
\hline
0 ~~ 0 ~~ 0 ~~~ 0 ~~~~~ 0 ~~ 0 ~~ 1 ~~w^2 ~~~~ 0 ~~ 0 ~~1 ~~w^2 ~~~~ 0 ~~ 0 ~~1 ~~ w^2 ~~~~0 ~~0~~~1  ~~w^2\\
0 ~~ 0 ~~ 0 ~~~ 1 ~~~~~ 0 ~~ 0 ~~ 0 ~~~1 ~~~~~ 0 ~~ 0 ~~0 ~~~1 ~~~~~ 0 ~~ 0 ~~0 ~~~ 1 ~~~~~0 ~~0~~~0  ~~~1~\\
0 ~~ 0 ~~ 1 ~~ w^2 ~~~~ 0 ~~ 0 ~~ 0 ~~~0 ~~~~~ 0 ~~ 0 ~~1 ~~w^2 ~~~~ 0 ~~ 0 ~~w ~~  1 ~~~~~0 ~~0~~w^2 ~~w~\\

\end{array}\right),
$$

where $\tilde{H}=\left(\begin{array}{cccc}
 1 & 0 & 1 & 1 \\
0 & 1 & 1 & w
\end{array}\right)$
is a generator matrix of the quaternary $[4, 2, 3]$  MDS code.
\end{construction}

\underline{\textbf{When $d=12$, $l-s=3$.}}

In this case, the optimal $(r, \delta)$-LRC has an equivalent parity-check matrix $H$ in the following form:
$$
H=\left(\begin{array}{cccc|cccc|cccc|c|cccc}
1 & 0 & 1 & 1 & 0 & 0 & 0 & 0 & 0 & 0 & 0 & 0 & \ldots  & 0 & 0 & 0 & 0\\
0 & 1 & 1 & w & 0 & 0 & 0 & 0 & 0 & 0 & 0 & 0 & \ldots  & 0 & 0 & 0 & 0\\
0 & 0 & 0 & 0 & 1 & 0 & 1 & 1 & 0 & 0 & 0 & 0 & \ldots  & 0 & 0 & 0 & 0\\
0 & 0 & 0 & 0 & 0 & 1 & 1 & w & 0 & 0 & 0 & 0 & \ldots  & 0 & 0 & 0 & 0 \\
0 & 0 & 0 & 0 & 0 & 0 & 0 & 0 & 1 & 0 & 1 & 1 & \ldots  & 0 & 0 & 0 & 0\\
0 & 0 & 0 & 0 & 0 & 0 & 0 & 0 & 0 & 1 & 1 & w & \ldots  & 0 & 0 & 0 & 0 \\
\vdots & \vdots & \vdots & \vdots & \vdots & \vdots & \vdots & \vdots & \vdots & \vdots & \vdots & \vdots & \ldots & \vdots & \vdots & \vdots & \vdots \\
0 & 0 & 0 & 0 & 0 & 0 & 0 & 0 & 0 & 0 & 0 & 0 & \ldots  &1 & 0 & 1 & 1 \\
0 & 0 & 0 & 0 & 0 & 0 & 0 & 0 & 0 & 0 & 0 & 0 & \ldots  &0 & 1 & 1 & w \\
\hline
\textbf{0} & \textbf{0} & \textbf{u}_{1} & \textbf{v}_{1} & \textbf{0} &\textbf{0} & \textbf{u}_{2}& \textbf{v}_{2}& \textbf{0} &\textbf{0} & \textbf{u}_{3}& \textbf{v}_{3}&\ldots & \textbf{0} &\textbf{0}& \textbf{u}_{l} &  \textbf{v}_{l}
\end{array}\right)
$$
where $\textbf{u}_{i},~ \textbf{v}_{i} \in \F_4^5$, $i \in [l]$.
Denote $\mathcal{V}_{ij} = \mbox{Span}\{\textbf{u}_i, \textbf{v}_i, \textbf{u}_j, \textbf{v}_j\}$, $i \neq j \in [l]$, the subspace of $\F_4^5$ spanned by $\textbf{u}_{i}$, $ \textbf{v}_{i}$, $\textbf{u}_{j}$ and $ \textbf{v}_{j}$. Since the minimum distance  $d=12$, it's also easy to verify that these vectors satisfy the following two properties:
\begin{enumerate}
\item [(1)] dim$(\mathcal{V}_{ij})$=4;
\item [(2)] for any $h \neq i \neq j$, $\textbf{u}_h$, $\textbf{v}_h$, $\textbf{u}_h + \textbf{v}_h$, $w\textbf{u}_h + \textbf{v}_h \notin \mathcal{V}_{ij}$.

\end{enumerate}

Consider the vectors $\textbf{u}_{i},~ \textbf{v}_{i}$, $i \in [l]$, as points in the $PG(4,\F_4)$,
let $L_i$ be the line generated by $\textbf{u}_{i}$ and $ \textbf{v}_{i}$. For any $h \neq i \neq j$, $L_h$ and $\mathcal{V}_{ij}$ intersect precisely at one point, which is denoted as $\textbf{p}_{ij}^h$. From above properties, we have $\textbf{p}_{ij}^h=w^2\textbf{u}_h + \textbf{v}_h$, equivalently, $\textbf{p}_{hj}^i=w^2\textbf{u}_i + \textbf{v}_i$, $\textbf{p}_{hi}^j=w^2\textbf{u}_j + \textbf{v}_j$.
Since $L_h \subset \mathcal{V}_{hi} \cap \mathcal{V}_{hj}$, $\textbf{p}_{hj}^i \in \mathcal{V}_{hi} \cap \mathcal{V}_{hj}$, $\textbf{p}_{hj}^i \in \mathcal{V}_{hi} \cap \mathcal{V}_{hj}$, the dimension of $\mathcal{V}_{hi} \cap \mathcal{V}_{hj}$ is 3, we have the dimension of the subspace Span$\{L_h,\textbf{p}_{hj}^i,\textbf{p}_{hj}^i\}$ is 3. This means that two lines $L_h$ and $L_{\textbf{p}_{hj}^i,\textbf{p}_{hj}^i}$ intersect at a point, and further, that intersection point is $\textbf{p}_{ij}^h$, where $L_h$ intersects $\mathcal{V}_{ij}$. Therefore, the points $\textbf{p}_{ij}^h$, $\textbf{p}_{hj}^i$ and $\textbf{p}_{hj}^i$ are collinear, that is, all the points $w^2\textbf{u}_i + \textbf{v}_i$, $i \in [l]$, are collinear.
There are at most five points in a given line, therefore $l \leq 5$. Since $l=s+3$, $s\geq 2$, we get $l=5$, $n=20$.

In this case, the optimal $(r, \delta)$-LRC code has parameters
\begin{equation}\label{l-s=3_1}
n =20,~ k= 3,~ d=12,~ r = 2, ~\delta =3.
\end{equation}

\begin{construction}\label{constructionl-s=3_1}
The following parity-check matrix $H$ gives the optimal quaternary $(r, \delta)$-LRC code for the parameters in (\ref{l-s=3_1}).

$$
H=\left(\begin{array}{c}
I_{5} \otimes \tilde{H} \\
\hline
0 ~~ 0 ~~ 0 ~~~ 0 ~~~~~ 0 ~~ 0 ~~ 0 ~~~0 ~~~~~ 0 ~~ 0 ~~1 ~~w^2 ~~~~ 0 ~~ 0 ~~1 ~~ w^2 ~~~~0 ~~0~~~1 ~~w^2\\
0 ~~ 0 ~~ 0 ~~~ 0 ~~~~~ 0 ~~ 0 ~~ 0 ~~~1 ~~~~~ 0 ~~ 0 ~~0 ~~~1 ~~~~~ 0 ~~ 0 ~~1 ~~ w ~~~~~0 ~~0~~~w ~~~0~\\
0 ~~ 0 ~~ 0 ~~~ 0 ~~~~~ 0 ~~ 0 ~~ 1 ~~w^2 ~~~~ 0 ~~ 0 ~~0 ~~~0 ~~~~~ 0 ~~ 0 ~~1 ~~ w^2 ~~~~0 ~~0~~~w ~~~1~\\
0 ~~ 0 ~~ 0 ~~~ 1 ~~~~~ 0 ~~ 0 ~~ 0 ~~~0 ~~~~~ 0 ~~ 0 ~~0 ~~~1 ~~~~~ 0 ~~ 0 ~~1 ~~~ 1 ~~~~~0 ~~0~~w^2 ~~1~\\
0 ~~ 0 ~~ 1 ~~ w^2 ~~~~ 0 ~~ 0 ~~ 0 ~~~0 ~~~~~ 0 ~~ 0 ~~1 ~~w^2 ~~~~ 0 ~~ 0 ~~1 ~~ w^2 ~~~~0 ~~0~~w^2 ~~w~\\
\end{array}\right),
$$

where $\tilde{H}=\left(\begin{array}{cccc}
 1 & 0 & 1 & 1 \\
0 & 1 & 1 & w
\end{array}\right)$
is a generator matrix of the quaternary $[4, 2, 3]$  MDS code.
\end{construction}

\subsubsection{Case $r=2,~\delta=4$.}

We have $k=2s+1,$ the codes $\langle H_{i}' \rangle $, $i\in[l]$, are $[5,3,3]$ MDS codes, the codes $\mathcal C|_{S_i}$ are $[5,2,4]$ MDS codes.

\underline{\textbf{If $s=1$.}} $n=5+d$, $k = 3$ and $5 \leq d\leq  4\delta=16$.

When $d=16$, the optimal LRC code $\mathcal C$ is a $[21,3,16]$ linear code with locality $(2, 4)$. As in the case \ref{5.2.1}, we can think of the columns of generator matrix $G$ as points in the projective plane $PG(2,\F_4)$.
We want to find $21$ points from $PG(2,\F_4)$ to form matrix $G$. To ensure a locality equals $(r=2, \delta=4)$ for the code $\mathcal C$, these points need to be such that each point is on a line consisting of $r+\delta-1=5$ points. Besides, to ensure minimum distance $d=16$, these points need to be such that any $n-d+1=6$ points are not on the same line. It is easy to see that the matrix $G$ formed by taking all the $21$ points in $PG(2,\F_4)$ as columns is the generator matrix of code $\mathcal C$ we want. Similarly, according to the principle, we can also obtain the generation matrix of code $\mathcal C$ of all possible parameters by selecting appropriate points.

\begin{Claim}\label{claim1}
In this case, the minimum distance  $d$ cannot be $5$ or $6$, that is, there is no the optimal $(2, 4)$-LRC code with parameters $[10,3,5]$ or $[11,3,6]$ over $\F_4$.
\end{Claim}
\begin{proof}
When $d=5$, the optimal $(2, 4)$-LRC code $\mathcal C$ is a $[10,3,5]$ linear code. From Proposition \ref{l}, the number of local groups in parity-check matrix $H$ is $l=2$. Since all local groups $H_i$, $i\in[2]$, have supports size $\#S_i=5$, the union of the supports of all local groups $\#(S_1\cup S_2) = n=10$, we have the supports of these two local groups are disjoint.
After removing $\lceil\frac{k}{r}\rceil-1=1$ local group from $H$, without loss of generality, suppose $H_{1}'$ is removed, then matrix $H'$ contains $H_2'$ as its submatrix, that is, a $[5,4,2]$ MDS code $C_1$ contains a $[5,3,3]$ MDS code $C_2$ as its subcode. We say this situation is impossible. If $C_2$ is a subcode of $C_1$, then $C_1^{\bot}$ is a subcode of $C_2^{\bot}$. Since the weight distribution of MDS code can be uniquely determined by its code parameters \cite{MacWilliams1977}, $C_2^{\bot}$ is a $[5,2,4]$ MDS code, we can get there is no codeword of weight 5 in code $C_2^{\bot}$.  Finally note that $C_1^{\bot}$ is a $[5,1,5]$ MDS code, it cannot be  be a subcode of $C_2^{\bot}$.

Now, when $d=6$, the optimal $(2, 4)$-LRC code $\mathcal C$ is a $[11,3,6]$ linear code. From Proposition \ref{l}, we can also get the number of local groups in parity-check matrix $H$ is $l=2$. Since all local groups $H_i$, $i\in[2]$, have supports size $\#S_i=5$, we have the union of the supports of these two local groups $\#(S_1\cup S_2)< n=11$, this contradicts the property of the local group.
\end{proof}

In this case, the optimal $(r, \delta)$-LRC code has parameters
\begin{equation}\label{17}
n = d+5 ,~ k= 3,~ 7 \leq d\leq 16,~ r = 2, ~\delta =4.
\end{equation}

Let matrix
$$
G=\left(\begin{array}{llllllllllllllllllllllllll}
0& & 0 &0 &0  &0 & &1  &1&1&1 & &1  &1&1&1 & &1  &1&1&w^2 & &1&1&w&1\\
1& & 0 &1 &1  &1 & &w^2&1&w&0 & &w^2&w&0&1 & &w^2&0&1&1   & &0&1&1&w\\
0& & 1 &w &w^2&1 & &0  &0&0&0 & &1  &1&1&1 & &w  &w&w&1   & &w^2&w^2&1&w^2
\end{array}\right)_{3\times 21}.
$$
In fact, the columns of the matrix $G$ (denoted by $\textbf{g}_j$, $j\in[21]$) can correspond in turn to points $\textbf{a}$, $\textbf{a}_{1i}$, $\textbf{a}_{2i}$, $\textbf{a}_{3i}$, $\textbf{a}_{4i}$ and $\textbf{a}_{5i}$, $i\in[4]$, in Figure 1.


\begin{construction}\label{construction17}

The above generator matrix $G$ gives the optimal quaternary $(r, \delta)$-LRC code $\mathcal C$ for the parameters in (\ref{17}) with $d=16$. By puncturing $\mathcal C$ on the coordinates set $S$, with $S=\{2\}$, $\{2,3\}$, $\{2,3,4\}$, $\{2,...,5\}$, $\{13,17,18,19,21\}$, $\{15,...,19,21\}$, $\{12,13,16,..,19,21\}$, $\{14,15,...,21\}$, $\{11,12,13,15,...,19,21\}$, we can get the optimal quaternary $(r, \delta)$-LRC codes with $d=15,~14,...,7$, respectively.

\end{construction}

\underline{\textbf{If $s\geq 2$.}} $n=5s+d$, $k=2s+1$ and $5 \leq d \leq 16$.
The supports of all $l$ local groups are disjoint and $n=l(r+\delta-1)=5l$.
From $5s+d=5l$, we get $5 \mid d$, $d=5$ 10 or $15$, $l-s=\frac{d}{5}=1,~2$ or $3$.

\underline{\textbf{When $d=5$, $l-s=1$.}}

Each $\langle H_{i}' \rangle $, $i\in[l]$, is a $[5,3,3]$ MDS code and $\langle H' \rangle $ is a $[5,4,2]$ MDS code. From the proof of Claim \ref{claim1}, we know that this situation  cannot happen.

\underline{\textbf{When $d=10$, $l-s=2$.}}

In this case, the optimal $(r, \delta)$-LRC has an equivalent parity-check matrix $H$ in the following form:
$$
H=\left(\begin{array}{ccccc|ccccc|c|ccccc}
1 & 0 & 0 & 1 & 1   & 0 & 0 & 0 & 0 & 0   & \ldots  & 0 & 0 & 0 & 0& 0\\
0 & 1 & 0 & 1 & w   & 0 & 0 & 0 & 0 & 0   & \ldots  & 0 & 0 & 0 & 0& 0\\
0 & 0 & 1 & 1 & w^2 & 0 & 0 & 0 & 0 & 0   & \ldots  & 0 & 0 & 0 & 0& 0\\
0 & 0 & 0 & 0 & 0   & 1 & 0 & 0 & 1 & 1   & \ldots  & 0 & 0 & 0 & 0& 0\\
0 & 0 & 0 & 0 & 0   & 0 & 1 & 0 & 1 & w   & \ldots  & 0 & 0 & 0 & 0& 0\\
0 & 0 & 0 & 0 & 0   & 0 & 0 & 1 & 1 & w^2 & \ldots  & 0 & 0 & 0 & 0& 0 \\
\vdots & \vdots & \vdots & \vdots & \vdots & \vdots & \vdots & \vdots & \vdots & \vdots & \ldots & \vdots & \vdots & \vdots & \vdots & \vdots\\
0 & 0 & 0 & 0 & 0   & 0 & 0 & 0 & 0 & 0 & \ldots  &1 & 0 & 0 & 1 & 1 \\
0 & 0 & 0 & 0 & 0   & 0 & 0 & 0 & 0 & 0 & \ldots  &0 & 1 & 0 & 1 & w \\
0 & 0 & 0 & 0 & 0   & 0 & 0 & 0 & 0 & 0 & \ldots  &0 & 0 & 0 & 1 & w^2 \\
\hline
\textbf{0} & \textbf{0} & \textbf{0} & \textbf{u}_{1} & \textbf{v}_{1} & \textbf{0} & \textbf{0} &\textbf{0} & \textbf{u}_{2}& \textbf{v}_{2}& \ldots & \textbf{0} &\textbf{0}& \textbf{0} & \textbf{u}_{l} &  \textbf{v}_{l}
\end{array}\right)
$$
where $\textbf{u}_{i},~ \textbf{v}_{i} \in \F_4^3$, $i \in [l]$.
Denote $\mathcal{V}_i = \mbox{Span}\{\textbf{u}_i, \textbf{v}_i\}$ the subspace of $\F_4^3$ spanned by $\textbf{u}_{i}$ and $ \textbf{v}_{i}$. Similarly, these vectors satisfy the following two properties:
\begin{enumerate}
\item [(1)] dim$(\mathcal{V}_i)$=2;
\item [(2)] for any $j \neq i \in [l]$, $\textbf{u}_i$, $\textbf{v}_i$, $\textbf{u}_i + \textbf{v}_i$, $w\textbf{u}_i + \textbf{v}_i$, $w^2\textbf{u}_i + \textbf{v}_i \notin \mathcal{V}_j$.

\end{enumerate}

 Consider the vectors $\textbf{u}_{i},~ \textbf{v}_{i}$, $i \in [l]$, as points in $PG(2,\F_4)$, then $\mathcal{V}_i$ is a line in $PG(2,\F_4)$ generated by $\textbf{u}_{i}$ and $ \textbf{v}_{i}$.
The second property requires that for any $i \neq j \in [l]$, $\mathcal{V}_i$ and $\mathcal{V}_j$ should have no intersection, which contradicts the fact that any two lines in $PG(2,\F_4)$ must intersect.
Therefore, there is no optimal $(2, 4)$-LRC code in this case.






\underline{\textbf{When $d=15$, $l-s=3$.}}

In this case, the optimal $(r, \delta)$-LRC has an equivalent parity-check matrix $H$ in the following form:
$$
H=\left(\begin{array}{ccccc|ccccc|ccccc|c|ccccc}
1 & 0 & 0 & 1& 1   & 0 & 0 & 0 & 0 & 0 & 0 & 0 & 0 & 0 & 0& \ldots  & 0 & 0 & 0 & 0& 0\\
0 & 1 & 0 & 1& w   & 0 & 0 & 0 & 0 & 0 & 0 & 0 & 0 & 0 & 0& \ldots  & 0 & 0 & 0 & 0& 0\\
0 & 0 & 1 & 1& w^2 & 0 & 0 & 0 & 0 & 0 & 0 & 0 & 0 & 0 & 0& \ldots  & 0 & 0 & 0 & 0& 0\\
0 & 0 & 0 & 0 & 0  & 1 & 0 & 0 & 1& 1   & 0 & 0 & 0 & 0 & 0 & \ldots  & 0 & 0 & 0 & 0& 0 \\
0 & 0 & 0 & 0 & 0  & 0 & 1 & 0 & 1& w   & 0 & 0 & 0 & 0 & 0 & \ldots  & 0 & 0 & 0 & 0& 0\\
0 & 0 & 0 & 0 & 0  & 0 & 0 & 1 & 1& w^2 & 0 & 0 & 0 & 0 & 0 & \ldots  & 0 & 0 & 0 & 0& 0\\
0 & 0 & 0 & 0 & 0  & 0 & 0 & 0 & 0 & 0  & 1 & 0 & 0 & 1 & 1 & \ldots  & 0 & 0 & 0 & 0& 0 \\
0 & 0 & 0 & 0 & 0  & 0 & 0 & 0 & 0 & 0  & 0 & 1 & 0 & 1 & w & \ldots  & 0 & 0 & 0 & 0& 0\\
0 & 0 & 0 & 0 & 0  & 0 & 0 & 0 & 0 & 0  & 0 & 0 & 1 & 1 & w^2 & \ldots  & 0 & 0 & 0 & 0& 0 \\
\vdots & \vdots & \vdots & \vdots &\vdots & \vdots & \vdots & \vdots & \vdots & \vdots & \vdots & \vdots & \vdots & \vdots & \vdots & \ldots & \vdots & \vdots & \vdots & \vdots & \vdots\\
0 & 0 & 0 & 0 & 0  & 0 & 0 & 0 & 0 & 0  & 0 & 0 & 0 & 0& 0 & \ldots& 1 & 0 & 0 & 1 & 1\\
0 & 0 & 0 & 0 & 0  & 0 & 0 & 0 & 0 & 0  & 0 & 0 & 0 & 0& 0 & \ldots& 0 & 1 & 0 & 1 & w\\
0 & 0 & 0 & 0 & 0  & 0 & 0 & 0 & 0 & 0  & 0 & 0 & 0 & 0& 0 & \ldots& 0 & 0 & 1 & 1 & w^2\\
\hline
\textbf{0} & \textbf{0} & \textbf{0}& \textbf{u}_{1} & \textbf{v}_{1} & \textbf{0}& \textbf{0} &\textbf{0} & \textbf{u}_{2}& \textbf{v}_{2}& \textbf{0}& \textbf{0} &\textbf{0} & \textbf{u}_{3}& \textbf{v}_{3}&\ldots & \textbf{0} &\textbf{0}& \textbf{0}& \textbf{u}_{l} &  \textbf{v}_{l}
\end{array}\right)
$$
where $\textbf{u}_{i},~ \textbf{v}_{i} \in \F_4^5$, $i \in [l]$.
Denote $\mathcal{V}_{ij} = \mbox{Span}\{\textbf{u}_i, \textbf{v}_i, \textbf{u}_j, \textbf{v}_j\}$, $i \neq j \in [l]$, the subspace of $\F_4^5$ spanned by $\textbf{u}_{i}$, $ \textbf{v}_{i}$, $\textbf{u}_{j}$ and $ \textbf{v}_{j}$. These vectors satisfy the following two properties:
\begin{enumerate}
\item [(1)] dim$(\mathcal{V}_{ij})$=4;
\item [(2)] for any $h \neq i \neq j$, $\textbf{u}_h$, $\textbf{v}_h$, $\textbf{u}_h + \textbf{v}_h$, $w\textbf{u}_h + \textbf{v}_h$, $w^2\textbf{u}_h + \textbf{v}_h \notin \mathcal{V}_{ij}$.

\end{enumerate}
Consider the vectors $\textbf{u}_{i},~ \textbf{v}_{i}$, $i \in [l]$, as points in the $PG(4,\F_4)$, let $L_i$ be the line generated by $\textbf{u}_{i}$ and $ \textbf{v}_{i}$.
The second property requires that for any $h \neq i \neq j$, $L_h$ and $\mathcal{V}_{ij}$ should have no intersection, which contradicts the fact that they must intersect in $PG(4,\F_4)$.
Therefore, there is no optimal $(2, 4)$-LRC code in this case.






\subsubsection{Case $r=3,~\delta=3$.}\label{r3delta3}

We have $k=3s+1,$ the codes $\langle H_{i}' \rangle $, $i\in[l]$, are $[5,2,4]$ MDS codes, the codes $\mathcal C|_{S_i}$, are $[5,3,3]$ MDS codes.

\underline{\textbf{If $s=1$.}} $n=5+d$, $k = 4$ and $5 \leq d\leq 12$.

The optimal LRC code $\mathcal C$ is a $[5+d,4,d]$ linear code with locality $(3, 3)$. We can think of the columns of generator matrix $G$ as points in the projective plane $PG(3,\F_4)$.
Fix a line, five different faces contain this line, we know that these five faces cover all $85$ points in $PG(3,\F_4)$.

When $d=12$, we want to find $17$ points from $PG(3,\F_4)$ to form matrix $G$. To ensure a locality $(r=3, \delta=3)$  for the code $\mathcal C$ is $(r=3, \delta=3)$, these points need to be such that each point is on a face group consisting of $r+\delta-1=5$ points, and these 5 points satisfy that any three points are not collinear. To ensure a minimum distance for $\mathcal C$ equals $d=12$, these points need to be such that any $n-d+1=6$ points are not on the same face. We get the following matrix $G$ that satisfies the conditions by selecting appropriate points.

$$
G=\left(\begin{array}{llllllllllllllllllllllllll}
0&0 && 0&0&0 && 1&1&1 && 1&1&1 && 1&1&1 && 1&1&1\\
0&0&& 1&1&1&& 0&0&0& &1&1&1& &w&w&w& &w^2&w^2&w^2\\
0&1&&1&w&w^2&&0&w&w^2&&0&w&w^2&&w&0&w^2&&w&w^2&0\\
1&1&&w&w^2&w^2&&0&1&1&&0&0&w&&1&1&w^2&&0&0&1

\end{array}\right).
$$


In this case, the optimal $(r, \delta)$-LRC code has parameters
\begin{equation}\label{18}
n = d+5 ,~ k= 4,~ 5\leq d\leq 12,~ r = 3, ~\delta =3.
\end{equation}


\begin{construction}\label{construction18}
The above generator matrix $G$ gives the optimal quaternary $(r, \delta)$-LRC code $\mathcal C$ for the parameters in (\ref{18}) with $d=12$. By puncturing $\mathcal C$ on the coordinates set $S$, with $S=\{17\}$, $\{16,17\}$, $\{15,16,17\}$, ... $\{11,12,...,17\}$, we can get the optimal quaternary $(r, \delta)$-LRC codes with $d=11,~10,~9,...,~5$, respectively.

\end{construction}

\underline{\textbf{If $s\geq 2$.}} $n=5s+d$, $k=3s+1$ and $5 \leq d \leq 12$.
The supports of all $l$ local groups are disjoint and $n=l(r+\delta-1)=5l$.
From $5s+d=5l$, we get $5 \mid d$, $d=5$ or 10, accordingly, $l-s=\frac{d}{5}=1$ or $2$.

\underline{\textbf{When $d=5$, $l-s=1$.}}

The optimal $(r, \delta)$-LRC code has parameters
\begin{equation}\label{l-s=1_3}
n = 5l,~ k =3l-2,~ d=5,~ r = 3, ~\delta =3,~(l\geq 3).
\end{equation}


\begin{construction}\label{constructionl-s=1_3}
The following parity-check matrix $H$ gives the optimal quaternary $(r, \delta)$-LRC code for the parameters in (\ref{l-s=1_3}).
$$
H=\left(\begin{array}{c}
 I_{l} \otimes \tilde{H}\\
 \hline
 \textbf{1}_{l} \otimes \left(\begin{array}{lllll}
 0&0&1&0&w^2\\
0 &0 & 0 & 1&w^2
\end{array}\right)
\end{array}\right),
$$

where $\tilde{H}= \left(\begin{array}{lllll}
1&0&1&1&1\\
0 &1 & 1 & w&w^2
\end{array}\right)$ is a generator matrix of the quaternary $[5, 2, 4]$ MDS code.
\end{construction}

\underline{\textbf{When $d=10$, $l-s=2$.}}

In this case, the optimal $(r, \delta)$-LRC has an equivalent parity-check matrix $H$ in the following form:
$$
H=\left(\begin{array}{ccccc|ccccc|c|ccccc}
1 & 0 & 1 & 1 & 1   & 0 & 0 & 0 & 0 & 0   & \ldots  & 0 & 0 & 0 & 0& 0\\
0 & 1 & 1 & w & w^2 & 0 & 0 & 0 & 0 & 0   & \ldots  & 0 & 0 & 0 & 0& 0\\
0 & 0 & 0 & 0 & 0   & 1 & 0 & 1 & 1 & 1   & \ldots  & 0 & 0 & 0 & 0& 0\\
0 & 0 & 0 & 0 & 0   & 0 & 1 & 1 & w & w^2 & \ldots  & 0 & 0 & 0 & 0& 0 \\
\vdots & \vdots & \vdots & \vdots & \vdots & \vdots & \vdots & \vdots & \vdots & \vdots & \ldots & \vdots & \vdots & \vdots & \vdots & \vdots\\
0 & 0 & 0 & 0 & 0   & 0 & 0 & 0 & 0 & 0 & \ldots  &1 & 0 & 1 & 1 & 1 \\
0 & 0 & 0 & 0 & 0   & 0 & 0 & 0 & 0 & 0 & \ldots  &0 & 1 & 1 & w & w^2 \\
\hline
\textbf{0} & \textbf{0} & \textbf{u}_{1} & \textbf{v}_{1} & \textbf{z}_1 &\textbf{0} & \textbf{0} & \textbf{u}_{2}& \textbf{v}_{2}&\textbf{z}_2 & \ldots & \textbf{0} &\textbf{0}& \textbf{u}_{l} &  \textbf{v}_{l}& \textbf{z}_l
\end{array}\right)
$$
where $\textbf{u}_{i},~ \textbf{v}_{i},~ \textbf{z}_i \in \F_4^5$, $i \in [l]$.

Denote $\mathcal{V}_i = \mbox{Span}\{\textbf{u}_i, \textbf{v}_i,\textbf{z}_i\}$ the subspace of $\F_4^5$ spanned by $\textbf{u}_{i}$, $ \textbf{v}_{i}$ and $\textbf{z}_i$, then dim$(\mathcal{V}_i)=3$.
Consider $\textbf{u}_{i},~ \textbf{v}_{i},~\textbf{z}_i$ as points in $PG(4,\F_4)$, then $\mathcal{V}_i$ contains $\frac{4^3-1}{4-1}=21$ different non-zero points in $PG(4,\F_4)$.

The dual of $[5,2,4]$ MDS code over $\F_4$ is $[5,3,3]$ MDS code $ C_1$ with weight distribution
$$A_0 = 1,~ A_1 = A_2 = 0,~ A_3 = 30,~ A_4 = 15, ~A_5 = 18.$$
Let $\mathcal C$ be the optimal $(r, \delta)$-LRC generated by above parity-check matrix $H$, then each codeword in code $\mathcal C$ is formed by splicing $l$ codewords in $ C_1$.
Let $\textbf{c}$ be a codeword of $C_1$, since for any $\xi \in \F_4^*$, vectors $(\textbf{0}~\textbf{0}~ \textbf{u}_{i}~\textbf{v}_{i} ~\textbf{z}_i) \cdot  (\xi \textbf{c})^{T}$ are the same point in $PG(4,\F_4)$, for convenience, we treat the codewords $\xi \textbf{c}$ in code $C_1$ as the same vector, and denote the resulting new set as $C_2$, then the weight distribution of $C_2$ is
$$A_0 = 1,~ A_1 = A_2 = 0,~ A_3 = 10,~ A_4 = 5, ~A_5 = 6.$$

It's easy to verify that there is a one-to-one correspondence between sets $\mathcal{V}_i$ and $C_2$, satisfying
$$\mathcal{V}_i=\{(\textbf{0}~\textbf{0}~ \textbf{u}_{i}~\textbf{v}_{i} ~\textbf{z}_i) \cdot \textbf{c}^{T}: \textbf{c} \in C_2\}.$$
Let $$\mathcal{V}_i(3)=\{(\textbf{0}~\textbf{0}~ \textbf{u}_{i}~\textbf{v}_{i} ~\textbf{z}_i) \cdot \textbf{c}^{T}: \textbf{c} \in C_2, wt(\textbf{c})=3\}.$$
For any $i_1,~i_2,~i_3 \in [l]$, fix a point $\textbf{p}$ in the set $\mathcal{V}_{i_1}(3)$. Then for any $\textbf{a}_1,~\textbf{a}_2 \in \mathcal{V}_{i_2}(3)$, these lines $L_{\textbf{p},\textbf{a}_1}$ and $L_{\textbf{p},\textbf{a}_2}$ are different, otherwise, there is a codeword with weight $\leq (3+5)$ in code $\mathcal C$. And for any $\textbf{a} \in \mathcal{V}_{i_2}(3), \textbf{b} \in \mathcal{V}_{i_3}(3)$, these lines $L_{\textbf{p},\textbf{a}}$ and $L_{\textbf{p},\textbf{b}}$ are different, otherwise, there is a codeword with weight $(3+3+3)$ in code $\mathcal C$.

Therefore, fix a point $\textbf{u}$ in the first set $\mathcal{V}_{1}(3)$, there are $10(l-1)$ different lines $L_{\textbf{u},\textbf{v}}$, $\textbf{v} \in \cup_{i=2}^{l} \mathcal{V}_{i}(3)$. In $PG(4,\F_4)$, the total number of points contained in these lines needs to satisfy $1+4\times 10(l-1)\leq 341-9$, that is $l \leq 9$. By $l=s+2$, $s\geq 2$, we get $4 \leq l \leq 9$.

In this case, the optimal $(r, \delta)$-LRC code has parameters

\begin{equation}\label{33d=10}
n = 5l ,~ k= 3l-5,~ d = 10,~ r = 3, ~\delta =3, (4 \leq l \leq 9).
\end{equation}

 However, we can only get the code length range in this case. It's quiet difficult to determine whether such an optimal code exists and its exact structure.

\subsubsection{Case $r=3,~\delta=4$.}

We have $k=3s+1,$ the codes $\langle H_{i}' \rangle $, $i\in[l]$, are $[6,3,4]$ MDS codes, the codes $\mathcal C|_{S_i}$, are $[6,3,3]$ MDS codes.

\underline{\textbf{If $s=1$.}} $n=6+d$, $k = 4$ and $5 \leq d\leq 16$.

From \cite{codetable}, we know that when $d \geq 13$, there is no quaternary linear code with parameters $[6+d,4,d]$.

When $d=12$, the optimal LRC code $\mathcal C$ is a $[18,4,12]$ linear code with locality $(3, 4)$. We can think of the columns of generator matrix $G$ as points in the projective plane $PG(3,\F_4)$.
we want to find $18$ points from $PG(3,\F_4)$ to form matrix $G$. To ensure the locality  $(r=3, \delta=4)$  for the code $\mathcal C$, these points need to be such that each point is on a face group consisting of $r+\delta-1=6$ points, and these 6 points satisfy that any three points are not collinear. To ensure minimum distance for $\mathcal C$ equals $d=12$, these points need to be such that any $n-d+1=7$ points are not on the same face. We get the following matrix $G$ that satisfies the conditions by selecting appropriate points.

$$
G=\left(\begin{array}{llllllllllllllllllll}
0& 0 & 0& 0&  0 & 0 & & 1&   w &  1 &  w &  1 &  w & &w^2& 1 & w^2& 1 & w^2&  1\\
1& 0 & 0& 1&  1 & 1 & & 1&   0 &  0 &  1 &  1 &  1 & &1&   0 &  0 & 1 &  1 &  1 \\
0& 1 & 0& 1&w^2 & w & & 0&   1 &  0  & 1 &w^2 &  w & &0&   1 &  0 & 1 &w^2 &  w \\
0& 0 & 1& 1&  w &w^2& & 0&   0 &  1  & 1&   w &w^2 & &0&   0 &  1 & 1 & w  &w^2

\end{array}\right).
$$

\begin{Claim}\label{claim2}
In this case, the minimum distance $d$ cannot be equals to $5$, that is, there is no the optimal $(3, 4)$-LRC code with parameters $[11,4,5]$ over $\F_4$.
\end{Claim}
\begin{proof}
When $d=5$, the optimal $(3, 4)$-LRC code $\mathcal C$ is a $[11,4,5]$ linear code. From Proposition \ref{l}, the number of local groups in parity-check matrix $H$ is $l=2$. Since each local group $H_i$, $i\in[2]$, has support size $\#S_i=6$, we have the supports of these two local groups intersect at one coordinate.
After removing $\lceil\frac{k}{r}\rceil-1=1$ local group from $H$, without loss of generality, suppose $H_{1}'$ is removed, then matrix $H'$ contains $\bar{H_2'}$ as its submatrix, where $\bar{H_2'}$ a is a submatrix obtained by deleting one column from matrix $H_2'$, that is, a $[5,4,2]$ MDS code $C_1$ contains a $[5,3,3]$ MDS code $C_2$ as its subcode. By the proof of Claim \ref{claim1}, we know this is impossible.

\end{proof}

In this case, the optimal $(r, \delta)$-LRC code has parameters
\begin{equation}\label{19}
n = d+6 ,~ k= 4,~ 6\leq d\leq 12,~ r = 3, ~\delta =4.
\end{equation}

\begin{construction}\label{construction19}
The above generator matrix $G$ gives the optimal quaternary $(r, \delta)$-LRC code $\mathcal C$ for the parameters in (\ref{18}) with $d=12$. By puncturing $\mathcal C$ on the coordinates set $S$, with $S=\{18\}$, $\{13,18\}$, $\{13,15,18\}$, $\{13,15,17,18\}$, $\{13,14,16,17,18\}$, $\{13,14,15,16,17,18\}$, we can get the optimal quaternary $(r, \delta)$-LRC codes with $d=11,~10,...,6$, respectively.

\end{construction}

\underline{\textbf{If $s\geq 2$.}} $n=6s+d$, $k=3s+1$ and $5 \leq d \leq 16$.
The supports of all $l$ local groups are disjoint and $n=l(r+\delta-1)=6l$.
From $6s+d=6l$, we get $6 \mid d$, $d=6$ or 12, accordingly, $l-s=\frac{d}{6}=1$ or $2$.

\underline{\textbf{When $d=6$, $l-s=1$.}}

The optimal $(r, \delta)$-LRC code has parameters

\begin{equation}\label{l-s=1_4}
n = 6l,~ k =3l-2,~ d=6,~ r = 3, ~\delta =4,~(l\geq 3).
\end{equation}

\begin{construction}\label{constructionl-s=1_4}
The following parity-check matrix $H$ gives the optimal quaternary $(r, \delta)$-LRC code for the parameters in (\ref{l-s=1_4}).
$$
H=\left(\begin{array}{c}
 I_{l} \otimes \tilde{H}\\
 \hline
 \textbf{1}_{l} \otimes \left(\begin{array}{llllll}
0&0&0&1&0&w^2\\
0&0&0&0&1&w^2
\end{array}\right)
\end{array}\right),
$$

where $\tilde{H}= \left(\begin{array}{llllll}
1&0&0&1&1&1\\
0&1&0&1& w&w^2\\
0&0&1&1& w^2&w
\end{array}\right)$ is the generator matrix of the quaternary $[6,3, 4]$ MDS code.
\end{construction}

\underline{\textbf{When $d=12$, $l-s=2$.}}

In this case, the optimal $(r, \delta)$-LRC has an equivalent parity-check matrix $H$ in the following form:
$$
H=\left(\begin{array}{cccccc|cccccc|c|cccccc}
1&0&0&1&1&1    & 0 & 0 & 0 & 0 & 0 &0  & \ldots  & 0 & 0 & 0 & 0& 0 &0\\
0&1&0&1& w&w^2 & 0 & 0 & 0 & 0 & 0 &0  & \ldots  & 0 & 0 & 0 & 0& 0 &0\\
0&0&1&1& w^2&w & 0 & 0 & 0 & 0 & 0 &0  & \ldots  & 0 & 0 & 0 & 0& 0 &0\\
0 & 0 & 0 & 0 & 0& 0 &  1&0&0&1&1&1     & \ldots  & 0 & 0 & 0 & 0& 0 &0\\
0 & 0 & 0 & 0 & 0& 0 &  0&1&0&1& w&w^2  & \ldots  & 0 & 0 & 0 & 0& 0 &0\\
0 & 0 & 0 & 0 & 0& 0  & 0&0&1&1& w^2&w  & \ldots  & 0 & 0 & 0 & 0& 0 &0\\
\vdots &\vdots &\vdots & \vdots & \vdots & \vdots & \vdots & \vdots & \vdots & \vdots & \vdots & \vdots & \ldots & \vdots & \vdots & \vdots &\vdots & \vdots & \vdots\\
0 & 0 & 0 & 0 & 0&0   & 0 & 0 & 0 & 0 & 0 &0& \ldots  &1&0&0&1&1&1 \\
0 & 0 & 0 & 0 & 0&0   & 0 & 0 & 0 & 0 & 0 &0& \ldots  &0&1&0&1& w&w^2 \\
0 & 0 & 0 & 0 & 0&0   & 0 & 0 & 0 & 0 & 0 &0& \ldots  &0&0&1&1& w^2&w \\
\hline
\textbf{0} & \textbf{0} &\textbf{0} & \textbf{u}_{1} & \textbf{v}_{1} & \textbf{z}_1 &\textbf{0} &\textbf{0} & \textbf{0} & \textbf{u}_{2}& \textbf{v}_{2}&\textbf{z}_2 & \ldots & \textbf{0} &\textbf{0}& \textbf{0} &\textbf{u}_{l} &  \textbf{v}_{l}& \textbf{z}_l
\end{array}\right)
$$
where $\textbf{u}_{i},~ \textbf{v}_{i},~ \textbf{z}_i \in \F_4^5$, $i \in [l]$.

Denote $\mathcal{V}_i = \mbox{Span}\{\textbf{u}_i, \textbf{v}_i,\textbf{z}_i\}$ the subspace of $\F_4^5$ spanned by $\textbf{u}_{i}$, $ \textbf{v}_{i}$ and $\textbf{z}_i$. Similarly, these vectors satisfy the following two properties:
\begin{enumerate}
\item [(1)] dim$(\mathcal{V}_i)$=3;
\item [(2)] for any $j \neq i \in [l]$, $a\textbf{u}_i + b\textbf{v}_i + c\textbf{z}_i \notin \mathcal{V}_j$, where $(a,b,c)\in PG(2,\F_4)$, $(a,b,c)=(1,w,w^2),~(1,w^2,w),~(1,1,1)$ or $wt((a,b,c))\leq 2$.

\end{enumerate}

The vectors $\textbf{u}_{i},~ \textbf{v}_{i},~\textbf{z}_i$, $i \in [l]$ are viewed as points in $PG(4,\F_4)$.
By the first property, for any $i \neq j \in [l]$, we have that the intersection of $\mathcal{V}_i$ and $\mathcal{V}_j$ is a subspace with dimension $\geq 1$ in $PG(4,\F_4)$.
By the second property, we know that the dimension of their intersection cannot be $\geq 2$. That is, $\mathcal{V}_i$ and $\mathcal{V}_j$ intersect at exactly one point in $PG(4,\F_4)$.

\begin{Claim}\label{claim3}
If all 3-dimensional subspace $\mathcal{V}_i$, $i \in [l]$, intersect at the same point, then $l\leq \frac{341-1}{21-1}=17$.
\end{Claim}

From the second property, a maximum of six points in a 3-dimensional subspace $\mathcal{V}_i$, $i \in [l]$, can be simultaneously contained in other 3-dimensional subspaces, which are $a\textbf{u}_i + b\textbf{v}_i + c\textbf{z}_i$ with $(a,b,c)=(1,1,w^2), ~(1,w^2,1), ~(1,w^2,w^2),$ $(1,1,w), ~(1,w,1)$ and $(1,w,w)$.

Therefore, in $PG(4,\F_4)$, $\mathcal{V}_1$ contains $21$ different non-zero points, $\mathcal{V}_1 \cup \mathcal{V}_2$ contains $21+20$ different non-zero points, $\mathcal{V}_1 \cup \mathcal{V}_2 \cup \mathcal{V}_3$ contains at least $21+20+19$ different non-zero points. By analogy, $\cup_{i=1}^l \mathcal{V}_i$ contains at least
$$21+20+19+18+17+16+15(l-6) = 15l+21 $$
 different points. This number should satisfy $15l+21 \leq 341$, that is $l \leq 21$.

\begin{Claim}\label{claim4}
In this case, $l$ cannot be equals to $21$.
\end{Claim}
\begin{proof}
The dual of $[6,3,4]$ MDS code over $\F_4$ is still the $[6,3,4]$ MDS code $ C_1$ with weight distribution
$$A_0 = 1,~ A_1 = A_2 = A_3 = 0,~ A_4 = 45, ~A_5 = 0, ~A_6 = 18.$$
Let $\mathcal C$ be the optimal $(r, \delta)$-LRC generated by above parity-check matrix $H$, then each codeword in code $\mathcal C$ is formed by splicing $l$ codewords in $ C_1$.
For convenience, we treat the codewords $\xi \textbf{c}$, $\xi \in \F_4^*$, in code $C_1$ as the same vector, and denote the resulting new set as $C_2$, then the weight distribution of $C_2$ is
$$A_0 = 1,~ A_1 = A_2 = A_3 = 0,~ A_4 = 15, ~A_5 = 0, ~A_6 = 6.$$

Similarly, there is a one-to-one correspondence between sets $\mathcal{V}_i$ and $C_2$, satisfying
$$\mathcal{V}_i=\{(\textbf{0}~\textbf{0}~\textbf{0}~ \textbf{u}_{i}~\textbf{v}_{i} ~\textbf{z}_i) \cdot \textbf{c}^{T}: \textbf{c} \in C_2\}.$$
Let $$\mathcal{V}_i(4)=\{(\textbf{0}~\textbf{0}~\textbf{0}~ \textbf{u}_{i}~\textbf{v}_{i} ~\textbf{z}_i) \cdot \textbf{c}^{T}: \textbf{c} \in C_2, wt(\textbf{c})=4\},$$
$$\mathcal{V}_i(6)=\{(\textbf{0}~\textbf{0}~\textbf{0}~ \textbf{u}_{i}~\textbf{v}_{i} ~\textbf{z}_i) \cdot \textbf{c}^{T}: \textbf{c} \in C_2, wt(\textbf{c})=6\}.$$

For any $i\neq j \in [l]$, $\mathcal{V}_i(4) \cap \mathcal{V}_j(4)= \emptyset$, otherwise, there is a codeword with weight $(4+4)$ in code $\mathcal C$.
If $l=21$, then
\begin{equation}\label{26}
\#  (\cup_{i=1}^l \mathcal{V}_i(4))= 21 \times 15=315,~ \#  (\cup_{i=1}^l \mathcal{V}_i(6)) \leq 341-315=26.
\end{equation}

We define the degree of a point $\textbf{p}$ as the number of subspaces $\mathcal{V}_i$, $i\in [l]$, that contain this point.
By Claim \ref{claim3}, there are at least 2 different intersections between these $l$ subspaces.
Let the maximum degree of these intersections be $\lambda$, and let the $\lambda$ subspaces containing a point of degree $\lambda$ be $\mathcal{V}_{i_j}$, $j\in [\lambda]$,
then $\#  (\cup_{j=1}^\lambda \mathcal{V}_{i_j}(6)) =5 \lambda +1$, this number should satisfy $5 \lambda +1 \leq 26$, that is $\lambda \leq 5$.

If $\lambda = 5$, $\#  (\cup_{j=1}^5 \mathcal{V}_{i_j}(6)) =26$, but there must be a point in the next subset $\mathcal{V}_{i_6}$ that is not included in $\cup_{j=1}^5 \mathcal{V}_{i_j}(6)$, which would contradict (\ref{26}).
If $\lambda \leq 4$, $\#  (\cup_{i=1}^{21} \mathcal{V}_{i}(6)) \geq \frac{21 \times 6}{4} > 31$, which also contradicts (\ref{26}).
Therefore, $l\neq 21$.

\end{proof}

In this case, the optimal $(r, \delta)$-LRC code has parameters

\begin{equation}\label{34l=4}
n = 6l ,~ k= 3l-5,~ d = 12,~ r = 3, ~\delta =4, (4 \leq l \leq 20).
\end{equation}


From Claim \ref{claim3}, we can give the construction when $l=17$, for example, we can take:
$$
\begin{array}{ccccc}
\textbf{u}_{1}=(w^2~0~w~0~0)^{T},~~&&\textbf{v}_{1}=(0~w^2~w^2~0~0)^{T} ,~~~&&\textbf{z}_1=(w^2~w^2~0~0~0)^{T},~~~~\\
\textbf{u}_{2}=(w^2~0~0~w^2~w)^{T},&&\textbf{v}_{2}=(0~0~0~w^2~1)^{T} ,~~~~~&&\textbf{z}_2=(w^2~0~0~0~1)^{T},~~~~~~\\
\textbf{u}_{3}=(0~1~w~1~w)^{T},~~~~&&\textbf{v}_{3}=(1~0~1~0~1)^{T} ,~~~~~~~&&\textbf{z}_3=(0~w^2~w~w^2~w)^{T},~~~~\\
\textbf{u}_{4}=(0~1~w~w~w^2)^{T},~~&&\textbf{v}_{4}=(1~0~1~0~w)^{T} ,~~~~~~~&&\textbf{z}_4=(w^2~1~0~w~0)^{T},~~~~~~\\
\textbf{u}_{5}=(0~1~w~w^2~1)^{T},~~&&\textbf{v}_{5}=(1~0~1~0~w^2)^{T} ,~~~~~&&\textbf{z}_5=(w~1~1~w^2~w^2)^{T},~~~~\\
\textbf{u}_{6}=(1~1~1~0~w)^{T},~~~~&&\textbf{v}_{6}=(1~w~1~w^2~w^2)^{T} ,~~~&&\textbf{z}_6=(1~1~w~w^2~w)^{T},~~~~~~\\
\textbf{u}_{7}=(1~1~1~1~w^2)^{T},~~&&\textbf{v}_{7}=(1~w~1~1~w)^{T} ,~~~~~~~&&\textbf{z}_7=(1~1~w~w~0)^{T},~~~~~~~~\\
\textbf{u}_{8}=(1~1~1~w~0)^{T},~~~~&&\textbf{v}_{8}=(1~w~w^2~w~w)^{T} ,~~~~~&&\textbf{z}_8=(w^2~w~1~0~1)^{T},~~~~~~\\
\textbf{u}_{9}=(1~1~1~w^2~1)^{T},~~&&\textbf{v}_{9}=(1~w~w^2~0~1)^{T} ,~~~~~&&\textbf{z}_9=(w~1~w^2~1~0)^{T},~~~~~~\\
\textbf{u}_{10}=(1~1~1~0~w^2)^{T},~&&\textbf{v}_{10}=(1~1~w^2~w^2~w^2)^{T} ,&&\textbf{z}_{10}=(w^2~w~w^2~w~1)^{T},~\\
\textbf{u}_{11}=(1~1~1~1~w)^{T},~~~&&\textbf{v}_{11}=(1~w^2~1~0~w^2)^{T} ,~~&&\textbf{z}_{11}=(1~1~w^2~w~0)^{T},~~~\\
\textbf{u}_{12}=(1~1~1~w~1)^{T},~~~&&\textbf{v}_{12}=(1~w~1~w~w^2)^{T} ,~~~~&&\textbf{z}_{12}=(w~w~w^2~1~0)^{T},~~~\\
\textbf{u}_{13}=(1~1~1~w^2~0)^{T},~&&\textbf{v}_{13}=(1~1~w~w~w)^{T} ,~~~~~~&&\textbf{z}_{13}=(1~w~1~0~w)^{T},~~~~~\\
\textbf{u}_{14}=(1~1~1~0~1)^{T},~~~&&\textbf{v}_{14}=(1~w~w^2~w^2~w)^{T},~~&&\textbf{z}_{14}=(w^2~w~1~w~w)^{T},~~~~\\
\textbf{u}_{15}=(1~1~1~1~0)^{T},~~~&&\textbf{v}_{15}=(1~w~w^2~1~1)^{T} ,~~~&&\textbf{z}_{15}=(w~1~w^2~0~w)^{T},~~~~\\
\textbf{u}_{16}=(1~1~1~w~w^2)^{T},~&&\textbf{v}_{16}=(1~1~w^2~w^2~0)^{T} ,~&&\textbf{z}_{16}=(w^2~w~w^2~0~w^2)^{T},\\
\textbf{u}_{17}=(1~1~1~w^2~w)^{T},~&&\textbf{v}_{17}=(1~w~w^2~w~0)^{T} ,~~~&&\textbf{z}_{17}=(w~1~w^2~w~w^2)^{T}.~~\\
\end{array}
$$

By removing $i$ groups, $0 \leq i \leq 13$,  we can get all constructions corresponding to $4 \leq l \leq 17$.

When $18 \leq l \leq 20$, we think that such an optimal $(r, \delta)$-LRC exists, but it is difficult to give the exact construction.

\section{Conclusion}\label{Sec-Conclusions}

This paper focused on optimal quaternary $(r, \delta)$-LRC codes specifically. More Specifically, on those $[n,k,d]$-linear codes over $\mathbb{F}_{4}$  with locality $(r, \delta)$ which are simultaneously $r$-optimal and $d$-optimal (with minimum distance $d \geq \delta > 2$, and dimension $k > r \geq 1$). By adopting parity-check matrix and generator approaches employing several related ingredients (such as local group and global group) and using techniques from coding theory to puncturing or shortened codes, we succeeded to provide a complete classification of optimal quaternary $(r, \delta)$-LRC codes achieving the generalized Singleton upper bound (\ref{rd_Singleton}). Our study includes the enumeration of all the possible code parameters of optimal quaternary $(r, \delta)$-LRC codes. We used arguments from finite geometry in the projective spaces over $\F_4$ and related objects and derived all constructions of optimal codes for each possible code parameter via its explicit parity-check matrix. Compared to the recent literature of this context, our structural and classification results about those optimal quaternary $(r,\delta)$-LRC codes are complete and, in addition, obtained through original proofs-techniques different from those already used.

\end{document}